\documentclass[conference,10pt]{IEEEtran}
\usepackage{authblk}
\usepackage{cite}
\usepackage{graphicx}
\usepackage{subfigure}
\usepackage{amsmath}
\usepackage{array}
\usepackage{amssymb}
\usepackage{amsfonts}
\usepackage{graphicx}
\usepackage{epstopdf}
\usepackage{placeins}
\usepackage{setspace}
\usepackage{amscd}
\usepackage{multirow}
\usepackage{mathrsfs}
\usepackage{epsfig}
\usepackage{textcomp}
\usepackage{mathtools}

\newcommand\coolunder[2]{\mathrlap{\smash{\underbrace{\phantom{%
    \begin{matrix} #2 \end{matrix}}}_{\mbox{$#1$}}}}#2}

\newtheorem{definition}{Definition}
\newtheorem{theorem}{Theorem}

\newtheorem{example}{Example}

\DeclareMathOperator{\rank}{rank}
\input{epsf.sty}
\title{Generalized Index Coding Problem and Discrete Polymatroids}
\begin{document}
\author{Anoop Thomas and B. Sundar Rajan
}\affil{Dept. of ECE, IISc, Bangalore 560012, India, Email: $\lbrace$anoopt,bsrajan$\rbrace$@ece.iisc.ernet.in.}
\date{\today}
\maketitle
 \thispagestyle{empty}	
%%%%%%%
\begin{abstract}
The index coding problem has been generalized recently to accommodate receivers which demand functions of messages and which possess functions of messages. 
The connections between index coding and matroid theory have been well studied in the recent past. Index coding solutions were first connected to multi linear representation of matroids. For vector linear index codes discrete polymatroids which can be viewed as a generalization of the matroids was used. It was shown that a vector linear solution to an index coding problem exists if and only if there exists a representable discrete polymatroid satisfying certain conditions. In this work we explore the connections between generalized index coding and discrete polymatroids. The conditions that need to be satisfied by a representable discrete polymatroid for a generalized index coding problem to have a vector linear solution is established. From a discrete polymatroid we construct an index coding problem with coded side information and shows that if the index coding problem has a certain optimal length solution then the discrete polymatroid satisfies certain properties. From a matroid we construct a similar generalized index coding problem and shows that the index coding problem has a binary scalar linear solution of optimal length if and only if the matroid is binary representable.
\end{abstract}

\section{Introduction}
\label{Sec:Introduction}

The broadcast nature of the wireless medium is utilized by many applications such as multimedia content delivery, audio and video on-demand and ad-hoc wireless networking. The index coding problem introduced by Birk and Kol \cite{ISCO} aims to increase the throughput of wireless networks. The model considered in \cite{ISCO} involves a source which possesses a set of messages and a set of receivers which demand messages. Each receiver knows a subset of messages which is referred to as the side information. The source also knows the side information available to the receivers. It uses this knowledge to develop proper encoding techniques to satisfy the demands of the receivers at an increased throughput. An index code is a encoding scheme developed by the source to satisfy all the receivers. An encoding scheme with minimum number of transmissions which enables all the receivers to decode its demanded messages is referred to as an optimal index code.

Bar-Yossef \textit{et al.} \cite{ICSI} studied a special case of index coding problem and found that the length of the optimal linear index code is equal to the minrank of a related graph which is an NP-hard problem. Graph theory techniques were used to find the optimal index codes for certain class of index coding problems in \cite{BCSI} and \cite{OngHo}.% Many recent works in index coding finds feasible sub-optimal solutions using linear programming techniques. The techniques partitions the set of receivers using clique covers, partial clique covers and partitions of appropriate graphs [citations]. 

An instance of the conventional index coding problem involves a source which possesses all the messages and a set of receivers. Each receiver possesses a subset of messages called the side information or the \textit{Has-set} and demands another subset of messages called the \textit{Want-set}. The wireless broadcast channel is assumed to be noiseless. The source is aware of the messages possessed by each receiver and it aims to reduce the number of transmissions required to satisfy the demands of all the receivers. The conventional index coding has been generalized to functional index coding in \cite{GuptaR16a}. In a functional index coding problem, the \textit{Has-set} and the \textit{Want-set} of users contain functions of messages rather than subsets of messages. Note that the conventional index coding is a special case of the functional index coding problem. The problem with the \textit{Has-sets} being linear combinations of messages was studied in \cite{BCCSI},\cite{ICCSI} where it was called as index coding with coded side information. This was motivated by the fact that certain clients may fail to receive some coded transmissions possibly due to power outage. The clients will now possess few coded transmissions as side information and the new problem is an index coding with coded side information. Dai \textit{et al.} \cite{dai2014data} considered both the \textit{Has-sets} and \textit{Want-sets} to be linear combinations of the messages which is referred to as generalized index coding problem (GIC). 

The connection between multi-linear representation of matroids and index coding was studied in \cite{ICMT}. It was shown in \cite{LICDPM} that a vector linear solution to an index coding problem exists if and only if there exists a representable discrete polymatroid satisfying certain conditions which are determined by the index coding problem. In this work we explore the connections between the generalized index coding and discrete polymatroids. The major contributions of this paper are as follows. 

\begin{itemize}
\item We establish a connection between vector linear index code for a generalized index coding problem and a representable discrete polymatroid in Section \ref{Sec:GICDPM}. It is shown that the existence of a linear solution for a generalized index coding problem is connected to the existence of a representable discrete polymatroid satisfying certain conditions determined by the generalized index coding problem.
\item From a discrete polymatroid we construct a generalized index coding problem and show that if the generalized index coding problem has a vector linear solution of optimal length over the binary field then the discrete polymatroid is representable over the binary field. An example to illustrate that the converse of the above result is not true is also provided.
\item A generalized index coding problem is constructed from matroids and it is shown that the constructed problem has a binary scalar linear solution if and only if the matroid is binary representable. Also, it is shown that certain generalized index coding problems do not have a binary scalar linear solution of optimal length using the above result.
\end{itemize}

The organization of the paper is as follows. In Section \ref{Sec:FIC} we review the definitions of functional index coding. In Section \ref{Sec:MatandDPM}, basic results of matroids and discrete polymatroids are reviewed. In Section \ref{Sec:GICDPM} the connections between generalized index coding and discrete polymatroids are established. In Section \ref{Sec:GICfromDPM} a generalized index coding problem is constructed from discrete polymatroids and it is shown that the index coding problem constructed has a vector linear solution only if the discrete polymatroid is representable. In Section \ref{Sec:MICCSI}, we construct an index coding with coded sided information problem from matroids and show that the constructed index coding problem has a binary scalar linear solution if and only if the matroid is binary representable. We conclude and summarize the results in Section \ref{Sec:Conclusion}.

\textit{Notations:} The set $\{1,2,\ldots,m \}$ is denoted as $\lceil m \rfloor$ and $\mathbb{Z}_{\geq 0}$ denote the set of non-negative integers. A vector of length $r$ whose $i^{\text{th}}$ component is one and all others components are zeros is denoted as $\epsilon_{i,r}$. For a vector $v$ of length $r$ and $A \subseteq \lceil r \rfloor, v(A)$ is the vector obtained by taking only the components of $v$ indexed by the elements of $A$. For $u,v \in \mathbb{Z}_{\geq 0}^r, u \leq v$ if all the components of $v-u$ are non-negative and $u<v$ if $u \leq v$ and $u \neq v$. For a set $S$, $\vert S \vert$ denotes the cardinality of the set $S$ and for a vector $v \in \mathbb{Z}_{\geq 0}^r, \vert v \vert$ denotes the sum of components of $v$. For $u,v \in \mathbb{Z}_{\geq 0}, u \vee v$ is the vector whose $i^{\text{th}}$ component is the maximum of the $i^{\text{th}}$ components of $u$ and $v$. For a vector $v \in \mathbb{Z}_{\geq 0}, (v)_{>0}$ denotes the set of indices corresponding to the non-zero components of $v$. For a matrix $M$, $M_{i}$ denotes the $i^{\text{th}}$ column of matrix $M$ and for a set $S$, $M_S$ denotes the submatrix obtained by concatenating the columns of $M$ indexed by the set $S$.

\section{Functional Index Coding}
\label{Sec:FIC}

An index coding problem $\mathcal{I}(X,\mathcal{R})$ includes 
\begin{itemize}
\item a set of messages $X=\lbrace x_1,x_2,\ldots,x_m \rbrace$ and
\item a set of receiver nodes $\mathcal{R} \subseteq \lbrace (x, H); x \in X, H \subseteq X \setminus \{x \} \}$.
\end{itemize}

For a receiver node $R=(x, H) \in \mathcal{R}$, $x$ denotes the message demanded by $R$ and $H$ denotes the side information possessed by $R$. Each one of the messages $x_i, i \in \lceil m \rfloor $ belongs to finite field $\mathbb{F}_q^n$.

An index code over $\mathbb{F}_q$ of length $l$ and dimension $n$ for the index coding problem $\mathcal{I}(X, \mathcal{R})$ is a function $f : \mathbb{F}_q^{mn} \rightarrow \mathbb{F}_q^l$, which satisfies the following condition. For every receiver $R=(x,H) \in \mathcal{R}$, there exists a function $\psi_R : \mathbb{F}_q^{n \vert H \vert + l} \rightarrow \mathbb{F}_q^n$ such that $\psi_R((x_i)_{i \in H},f(y))=x, \forall y \in \mathbb{F}_q^{mn}$. The function $\psi_R$ is referred to as the decoding function at receiver $R$. An index coding solution for which $n=1$ is called scalar solution and if $n > 1$ it is called a vector solution. An index code is called linear if the function $f$ is linear.

The index coding problem was generalized to functional index coding problem in \cite{GuptaR16a}. In functional index coding problem the side information and the demands of the receivers may be functions of messages rather than only a subset of the messages. The information possessed by the receivers is described by a \textit{Has-set} which consists of functions of messages. The demands of the receiver are described by a \textit{Want-set}. Each receiver $R_{i}$ is described by a tuple $(\mathcal{W}_{i}, \mathcal{H}_{i})$, where $\mathcal{W}_i,\mathcal{H}_{i}$ are sets of functions from $\mathbb{F}_{q}^{mn}$ to $\mathbb{F}_q$.

In this paper we consider those generalized index coding problems for which the functions demanded and possessed by the receivers are linear combinations of the messages. 

\begin{definition}
An instance $\mathcal{I}(X,\mathcal{R})$ of a generalized index coding problem comprises of\\
1) A source equipped with the message vector $X=(x_{1},x_{2},\ldots,x_{m})$, where $x_{i} \in \mathbb{F}_q^n, \forall ~ i \in \lceil m \rfloor$.  \\
2) A set of clients or receivers $\mathcal{R}= \lbrace R_{1},R_{2},\ldots,R_{\vert \mathcal{R} \vert} \rbrace$, where $R_{i}=(\mathcal{W}_i,\mathcal{H}_i)$ for all $R_{i} \in \mathcal{R}$. For any receiver $R_{i},\mathcal{H}_{i}= \lbrace h_{i,1}(X),h_{i,2}(X),\ldots,h_{i,\vert \mathcal{H}_{i} \vert }(X) \rbrace$ is the \textit{Has-set} where $h_{i,j} : \mathbb{F}_{q}^{mn} \rightarrow \mathbb{F}_{q}$ for $1 \leq j \leq \vert \mathcal{H}_{i} \vert$ and $\mathcal{W}_{i}=\lbrace w_{i,1}(X),w_{i,2}(X),\ldots,w_{i,\vert \mathcal{W}_{i} \vert }(X) \rbrace$ is the \textit{Want-set} where $w_{i,k} : \mathbb{F}_{q}^{mn} \rightarrow \mathbb{F}_{q}$ for $1 \leq k \leq \vert \mathcal{W}_{i} \vert$.
\end{definition}

Since the functions in the \textit{Has-set} of a receiver $R_{i}$ are linear it can be represented by vectors. Each function $h_{i,j} \in \mathcal{H}_{i}$ can be expressed as the inner product  $h_{i,j}(X)= X.K_{i,j}$ where $K_{i,j} \in \mathbb{F}_{q}^{mn}$. For the receiver $R_{i}$ we have $\vert \mathcal{H}_{i} \vert$ functions in the \textit{Has-set} each represented by a vector $K_{i,j}, 1 \leq j \leq \vert \mathcal{H}_{i} \vert$. All the functions in the \textit{Has-set} of receiver $R_{i}$ can be represented by a {\it knowledge matrix} $K_{i} \in \mathbb{F}_{q}^{mn \times \vert \mathcal{H}_{i} \vert}$. Note that $K_{i}=[K_{i,1},K_{i,2},\ldots,K_{i,\vert \mathcal{H}_{i} \vert}]$. Similarly the demand functions in  $\mathcal{W}_{i}$  can be represented by {\it demand vectors}.  Each function $w_{i,j} \in \mathcal{W}_{i}$ can be expressed as $w_{i,j}(X)= X.D_{i,j}$ where $D_{i,j} \in \mathbb{F}_{q}^{mn}$ and all the functions in the \textit{Want-set} of receiver $R_{i}$ can be described by the $mn \times \vert \mathcal{W}_{i} \vert$ {\it demand matrix} $D_{i}=[D_{i,1},D_{i,2},\ldots,D_{i,\vert \mathcal{W}_{i} \vert}]$.

%\begin{definition}(\cite{GuptaR16a} )
%An instance $\mathcal{I}(Z,\mathcal{R})$ of a functional index coding problem comprises \\
%1) a source equipped with the message vector $Z=(Z_{1},Z_{2}.\ldots,Z_{K})$, where, for every $k \in [K]$, $Z_{k}$ is uniformly distributed over $\mathbb{F}_{q}^{n_{k}}$ for some positive integer $n_{k}$ and  \\
%2) a set of clients or receivers $\mathcal{R}= \lbrace R_{1},R_{2},\ldots,R_{\vert \mathcal{R} \vert} \rbrace$, where $R_{i}=(H_{i},W_{i})$ for all $R_{i} \in \mathcal{R}$. For any receiver $R_{i},H_{i}= \lbrace h_{i,1}(Z),h_{i,2}(Z),\ldots,h_{i,\vert H_{i} \vert }(Z) \rbrace$ and $W_{i}=\lbrace w_{i,1}(Z),w_{i,2}(Z),\ldots, w_{i,\vert W_{i} \vert}(Z) \rbrace$ are the \textit{Has-} and \textit{Want-sets} respectively, where $h_{i,j},w_{i,l} : \mathbb{F}_{q}^{N_{K}} \rightarrow \mathbb{F}_{q}$ for $1 \leq j \leq \vert H_{i} \vert$ and $1 \leq l \leq \vert W_{i} \vert$, $(N_K=n_1+n_2+\ldots + n_K)$.
%\end{definition}
%
%Let $H_{i}(Z)=(h_{i,1}(Z),h_{i,2}(Z), \ldots h_{i,\vert H_{i} \vert} (Z))$ and $W_{i}(Z)=(w_{i,1}(Z),w_{i,2}(Z), \ldots w_{i,\vert W_{i} \vert} (Z))$. For a realization $z=(z_{1},z_{2},\ldots, z_{K})$ of $Z$ at the transmitter, $H_{i}(z)=(h_{i,1}(z),h_{i,2}(z), \ldots h_{i,\vert H_{i} \vert} (z))$ is the \textit{Has-value} known to $R_{i}$ and $W_{i}(z)=(w_{i,1}(z),w_{i,2}(z), \ldots w_{i,\vert W_{i} \vert} (z))$ is the \textit{Want-value} receiver $R_{i}$ is interested in computing. The conventional index coding problem correspond to the case in which all the \textit{Has-} and \textit{Want-set} are subsets of $Z$.

An index code over $\mathbb{F}_q$ of length $l$ and dimension $n$ for the generalized index coding problem $\mathcal{I}(X, \mathcal{R})$ is a function $f : \mathbb{F}_q^{mn} \rightarrow \mathbb{F}_q^l$, which satisfies the following condition. For every receiver $R_{i}=(\mathcal{W}_i,\mathcal{H}_{i}) \in \mathcal{R}$, there exists a function $\psi_{R_{i}} : \mathbb{F}_q^{ \vert \mathcal{H} \vert + l } \rightarrow \mathbb{F}_q^{\vert \mathcal{W}_{i} \vert}$ such that $\psi_{R_{i}}(X.K_{i},f(X))=X.D_{i}, \forall X \in \mathbb{F}_q^{mn}$. The definitions of linearity, scalar and vector index codes remains same as that of conventional index codes.

When the index code $f$ for a generalized index coding problem is linear it can be described as $f(X)=XL , \forall X \in \mathbb{F}_{q}^{mn}$, where $L$ is  a matrix of order $mn \times l$ over $\mathbb{F}_{q}$. The matrix $L$ is called as the matrix corresponding to the linear index code $f$ and the code $f$ is referred to as the linear index code based on $L$.

For an index coding problem $\mathcal{I}(X,\mathcal{R}),$ define
$\displaystyle{\mu (\mathcal{I}(X,\mathcal{R}))}$ as the maximum number of receivers having the same \textit{Has-set}. The length $l$ and dimension $n$ of an index coding solution for the index coding problem $\mathcal{I}(X,\mathcal{R})$ satisfy the condition $l/n \geq \mu (\mathcal{I}(X,\mathcal{R}))$ \cite{ICMT}. 
\begin{definition}[\cite{ICMT}]
An index coding solution for which $l/n = \mu (\mathcal{I}(X,\mathcal{R}))$ is defined to be a perfect index coding solution.
\end{definition}

\begin{example}
\label{Eg:GIC1}
Consider the generalized index coding problem with the message vector $X=[x_1 ~x_2 ~\ldots~x_5], x_{i} \in \mathbb{F}_2.$  There are five receivers $R_1=(x_1,\{x_2 \}),R_2=(x_2,\{x_1+x_5\}),R_3=(x_3, \{x_1,x_4\}),R_4=(x_4,\{x_1+x_2+x_3\})$ and $R_5=(x_5+x_4+x_3, \{ x_2, x_1+x_3 \})$. Consider receiver $R_5=(\mathcal{W}_5, \mathcal{H}_5). $ The knowledge matrix $K_5$ and demand matrix $D_5$ are as given below. 
 \begin{equation*}
   K_5=\begin{bmatrix}
          0 & 1 \\
          1 & 0\\
          0 & 1\\
          0 & 0 \\
          0 & 0
        \end{bmatrix},
   D_5=\begin{bmatrix}
          0 \\
          0 \\
          1 \\
          1 \\
          1
        \end{bmatrix}.
 \end{equation*} The source can satisfy the demands of all the receivers by transmitting three messages $x_1+x_2,x_3+x_4$ and $x_5$. The index code is linear and is described by the matrix  \begin{equation*}
   L=\begin{bmatrix}
          1 & 0 & 0 \\
          1 & 0 & 0 \\
          0 & 1 & 0 \\
          0 & 1 & 0 \\
          0 & 0 & 1
       \end{bmatrix}.
        \end{equation*}
\end{example}

\section{Matroids and Discrete Polymatroids}
\label{Sec:MatandDPM}

\subsection{Matroids}
\label{Sec:Matroids}
In this subsection we list few basic definitions and results from matroid theory. For a comprehensive treatment, the readers are referred to \cite{We,Ox}.

\begin{definition}
\label{def:Matroid}
Let $E$ be a finite set. A matroid $\mathcal{M}$ on $E$ is an ordered pair $(E,\mathcal{I})$, where the set $\mathcal{I}$ is a collection of subsets of $E$ satisfying the following three conditions
\begin{list}{}
\item (I1) $\phi \in \mathcal{I}$
\item (I2) If $X \in \mathcal{I}$ and $X' \subseteq X$, then $X' \in \mathcal{I}$.
\item (I3) If $X_{1}$ and $X_{2}$ are in $\mathcal{I}$ and $|X_{1}| < |X_{2}|$, then there is an element $e \in X_{2} - X_{1}$ such that $X_{1} \cup e \in \mathcal{I}$.
\end{list}
\end{definition}

The set $E$ is called the \textit{ground set} of the matroid and is also referred to as $E(\mathcal{M})$. The members of set $\mathcal{I}$ are called the independent sets of $\mathcal{M}$. Independent sets are also denoted by $\mathcal{I}(\mathcal{M})$. A maximal independent subset of $E$ is called a \textit{basis} of $\mathcal{M}$ and the set of all bases of $\mathcal{M}$ is denoted by $\mathcal{B}(\mathcal{M})$. A minimal dependent set $C \subseteq E$ is referred to as a \textit{circuit}. The set of all circuits of matroid $\mathcal{M}$ is denoted by $\mathfrak{C}(\mathcal{M})$. With $\mathcal{M}$, a function called the \textit{rank} function is associated, whose domain is the power set of $E$ and codomain is the set of non-negative integers. The rank of any $X \subseteq E$ in $\mathcal{M}$, denoted by $r_{\mathcal{M}}(X)$ is defined as the maximum cardinality of a subset $X$ that is a member of $\mathcal{I}(\mathcal{M})$.  The rank of matroid is the rank of its ground set.

The rank function of the matroid satisfies the following properties. 
\begin{itemize}
\item[(R1)] $r_{\mathcal{M}}(X) \leq \vert X \vert$, for all $X \subseteq E$.
\item[(R2)] $r_{\mathcal{M}}(X) \leq r_{\mathcal{M}}(Y)$, for all $X \subseteq Y \subseteq E$.
\item[(R3)] $r_{\mathcal{M}}(X \cup Y) + r_{\mathcal{M}}(X \cap Y) \leq r_{\mathcal{M}}(X) + r_{\mathcal{M}}(Y)$, for all $X,Y \subseteq E$.
\end{itemize}

Note that the rank of an independent set is equal to the cardinality of the independent set. A matroid is fully described by its rank function and a matroid $\mathcal{M}$ on ground set $E$ with rank function $r_{\mathcal{M}}$ is denoted as $\mathcal{M}(E,r_{\mathcal{M}})$.

A matroid $\mathcal{M}$ is said to be representable over $\mathbb{F}_q$ if there exists one-dimensional vector subspaces $V_1,V_2, \dotso V_{\vert E \vert}$ of a vector space $V$ such that $\dim(\sum_{i \in X} V_i)=r_{\mathcal{M}}(X), \forall X \subseteq E$ and the set of vector subspaces $V_i, i \in \lceil \vert E \vert  \rfloor,$ is said to form a representation of $\mathcal{M}.$ The one-dimensional vector subspaces $V_i, i \in \lceil \vert E \vert  \rfloor,$ can be described by a matrix $A$ over $\mathbb{F}_q$ whose $i^{\text{th}}$ column spans $V_i.$ 

A matroid $\mathcal{M}$ with matrix $A$ as its representation is called the vector matroid of $A$ and is denoted by $\mathcal{M}(A)$. Each element in the ground set of $\mathcal{M}(A)$ corresponds to a column in $A$. For a subset $S$ of ground set $E(\mathcal{M})$, $A_{S}$ denotes the submatrix of $A$ with columns corresponding to the elements of ground set in $S$.  

Multi-linear representation of matroids was introduced in \cite{SiAs,Ma}. 

\begin{definition}
\label{def:n-linearrepresentation}
A matroid $\mathcal{M}$ on the ground set $E$ is said to be multi-linearly representable of dimension $n$ over $\mathbb{F}_q$ if there exist vector subspaces $V_{1},V_{2},\ldots, V_{\vert E \vert}$ of a vector space $V$ over $\mathbb{F}_{q}$ such that $\dim(\sum_{i \in X} V_i)=nr_{\mathcal{M}}(X)), \forall X \subseteq E$. The vector subspaces $V_{1},V_{2}, \ldots, V_{\vert E \vert}$ are said to form a multi-linear representation of dimension $n$ over $\mathbb{F}_{q}$ for the matroid $\mathcal{M}$. The vector subspaces $V_{i}, i \in \lceil \vert E \vert \rfloor$ can be described by matrices $M_{1}, M_{2}, \ldots, M_{\vert E \vert} \in \mathbb{M}_{\mathbb{F}}(kn,n)$, where $k$ is the rank of the matroid. Let $M$ be the matrix obtained by concatenating the matrices $M_{1},M_{2},\ldots,M_{\vert E \vert}$, $M=[M_{1} ~ M_{2} ~\ldots~ M_{\vert E \vert}] $. For every subset $X \subseteq E$, rank$(M_X)=nr_\mathcal{M}(X)$.
\end{definition}

\subsection{Discrete Polymatroids}
\label{Sec:DPM}

In this subsection we review the definitions and results from discrete polymatroids. A discrete polymatroid $\mathbb{D}$ is defined as follows:
\begin{definition}[\cite{herzog2002discrete}]
A discrete polymatroid $\mathbb{D}$ on the ground set $\lceil m \rfloor$ is a non-empty finite set of vectors in $\mathbb{Z}_{\geq 0}^m$ satisfying the following conditions:
\begin{itemize}
\item
If $u \in \mathbb{D}$ and $v <u,$ then $v \in \mathbb{D}.$ 
\item
 For all $u, v \in \mathbb{D}$ with $\vert u\vert < \vert v\vert,$
there exists $w \in  \mathbb{D}$ such that $u < w  \leq  u \vee v.$
\end{itemize}
\end{definition}

Let $2^{\lceil m \rfloor}$ denote the power set of the set $\lceil m \rfloor$. For a discrete polymatroid $\mathbb{D},$ the rank function $\rho: 2^{\lceil m \rfloor} \rightarrow \mathbb{Z}_{\geq 0}$ is defined as $\rho(A)=\max \{ \vert u(A) \vert , u \in \mathbb{D}\},$ where $\emptyset \neq A  \subseteq \lceil m \rfloor$ and $\rho(\emptyset)=0.$ Alternatively, a discrete polymatroid $\mathbb{D}$ can be written in terms of its rank function as $\mathbb{D}=\lbrace x \in \mathbb{Z}_{\geq 0}^m: \vert x(A) \vert \leq \rho(A), \forall A \subseteq \lceil m \rfloor \rbrace.$ A discrete polymatroid is completely described by the rank function. So the discrete polymatroid $\mathbb{D}$ on $\lceil m \rfloor$ is also denoted by $(\lceil m \rfloor, \rho)$. The ground set of discrete polymatroid is also denoted by $E(\mathbb{D})$. 

A function $\rho: 2^{\lceil m \rfloor} \rightarrow \mathbb{Z}_{\geq 0}$ is the rank function of a discrete polymatroid if and only if it satisfies the following conditions \cite{farras2007ideal}:
\begin{description}
\item [(D1)]
For $A \subseteq B \subseteq \lceil m \rfloor,$ $\rho(A)\leq \rho(B).$
\item [(D2)]
 $\forall A,B \subseteq \lceil m \rfloor,$  $\rho(A \cup B) + \rho(A\cap B)\leq \rho(A)+\rho(B).$
\item [(D3)]
$\rho(\emptyset)=0.$
\end{description}

A vector $u \in  \mathbb{D}$ for which there does not exist $v \in \mathbb{D}$ such that $u<v,$ is called a basis vector of $\mathbb{D}.$  Let $\mathcal{B}(\mathbb{D})$ denote the set of basis vectors of $\mathbb{D}.$ The sum of the components of a basis vector of $\mathbb{D}$ is referred to as the rank of $\mathbb{D},$ denoted by $\rho(\mathbb{D}).$ Note that $\rho(\mathbb{D})=\rho(\lceil m \rfloor)$. For all the basis vectors, sum of the components will be equal \cite{vladoiu2006discrete}. A discrete polymatroid is nothing but the set of all integral subvectors of its basis vectors.

Consider a discrete polymatroid $\mathbb{D}$ with rank function $\rho$ on the ground set $\lceil m \rfloor$. Consider the function $\rho'(X)=n \rho(X), \forall X \subseteq \lceil m \rfloor$. The function $\rho'$ satisfies the conditions (D1),(D2) and (D3). The discrete polymatroid on the ground set $\lceil m \rfloor$ with the rank function $\rho'$ is denoted by $n\mathbb{D}$.

\begin{definition}[\cite{farras2007ideal}]
A discrete polymatroid $\mathbb{D}$ on the ground set $\lceil m \rfloor$ with rank function $\rho$ is said to be representable over $\mathbb{F}_q$ if there exists vector subspaces $V_1,V_2,\dotso,V_m$ of a vector space $E$ over $\mathbb{F}_q$ such that $dim(\sum_{i \in X} V_i)=\rho(X),$ $\forall X \subseteq \lceil m \rfloor.$ The set of vector subspaces $V_i,i\in\lceil m \rfloor,$ is said to form a representation of $\mathbb{D}.$ 
A discrete polymatroid is said to be representable if it is representable over some field. $\mathbb{D}(V_{1},V_{2}, \ldots, V_{m})$ denotes a representable discrete polymatroid on $\lceil m \rfloor$ with $V_{1},V_{2},\ldots,V_{m}$ as its representation. Each $V_{i}$ can be expressed as the column span of a $\rho(\lceil m \rfloor) \times \rho(\{i\})$ matrix $A_{i}$. The concatenated matrix $A=[A_{1}~A_{2} \ldots A_{m}]$ is referred to as the \textit{representing matrix} of the discrete polymatroid $\mathbb{D}$. 
\end{definition} 

\begin{definition} \cite{LICDPM}
For a discrete polymatroid $\mathbb{D}$ with rank function $\rho$ on the ground set $\lceil m \rfloor$, a vector $u \in \mathbb{Z}_{\geq 0}^{m}$ is said to be an excluded vector if the $i^{\text{th}}$ component of  $u$ is less than or equal to  $\rho(\{i\}), \forall i \in \lceil m \rfloor$ and $u \notin \mathbb{D}$. The set of excluded vectors for the discrete polymatroid $\mathbb{D}$ is denoted by $\mathcal{D}(\mathbb{D})$.
%\end{definition}
%\begin{definition}
 An excluded vector $u \in \mathcal{D}(\mathbb{D})$ is said to be a minimal excluded vector, if there does not exist $v \in \mathcal{D}(\mathbb{D})$ for which $v < u $. The set of minimal excluded vectors for the discrete polymatroid $\mathbb{D}$ is denoted by $\mathcal{C}(\mathbb{D})$.
\end{definition}

Discrete polymatroids can be viewed as a generalization of matroids \cite{LICDPM,herzog2002discrete,vladoiu2006discrete}. There is a one-to-one correspondence between the independent sets, basis sets, dependent sets and circuits of a matroid to the vectors of an associated discrete polymatroid. For a matroid $\mathcal{M}$ there is an associated discrete polymatroid $\mathbb{D}(\mathcal{M})$. Consider an independent set $I$ of the matroid $\mathcal{M}$. Corresponding to the set $I$ there exists a unique vector $\sum_{i \in I}\epsilon_{i,r} $ belonging to $\mathbb{D}(\mathcal{M})$. Discrete polymatroid $\mathbb{D}(\mathcal{M})$ can be written as $\{ \sum_{i \in I} \epsilon_{i,r} : I \in \mathcal{I}\}$ where $\mathcal{I}$ is the set of independent sets of matroid $\mathcal{M}$. 

For a basis set $B$ of a matroid $\mathcal{M}$, the vector $\sum_{i \in B}\epsilon_{i,r}$ is a basis vector of $\mathbb{D}(\mathcal{M})$ and for a basis vector $b$ of $\mathbb{D}(\mathcal{M})$, the set $(b)_{>0}$ is a basis set of $\mathcal{M}$. For a dependent set $D$ of $\mathcal{M}$, the vector $\sum_{i \in D} \epsilon_{i,r}$ is an excluded vector of $\mathbb{D}(\mathcal{M})$ and conversely for an excluded vector $d \in \mathcal{D}(\mathbb{D}(\mathcal{M}))$, the set $(d)_{>0}$ is a dependent set of $\mathcal{M}$. Similarly the set of minimal excluded vectors of $\mathbb{D}(\mathcal{M})$ and circuits of $\mathcal{M}$ are also related as follows. The set of circuits of matroid $\mathcal{M}$ is given by $\{(u)_>{0} : u \in \mathcal{C}(\mathbb{D}(\mathcal{M})) \}$. For a circuit $C$ of matroid $\mathcal{M}$ the vector $\sum_{i \in C} \epsilon_{i,r}$ is a minimal excluded vector for $\mathbb{D}(\mathcal{M})$.

\section{Generalized Index Coding Problem and Discrete Polymatroids}
\label{Sec:GICDPM}

In this section we explore the connections between generalized index coding problem and representable discrete polymatroids. Theorem \ref{thm:GICDPM} below connects the existence of a linear index code of length $l$ and dimension $n$ for a generalized index coding problem to the problem of representation of a discrete polymatroid satisfying certain conditions.  

\begin{theorem}
\label{thm:GICDPM}
A linear index code over $\mathbb{F}_{q}$ of length $l$ and dimension $n$ exists for a generalized index coding problem $\mathcal{I}(X,\mathcal{R})$ if and only if there exists a discrete polymatroid $\mathbb{D} = ( \lceil m +1 \rfloor, \rho)$ representable over $\mathbb{F}_{q}$ with $\rho(\mathbb{D})=mn$ and with $A_1,A_2,\ldots,A_{m+1},$ as the representation matrices satisfying the following conditions : \newline
(C1) $\rho(\{i\})=n, \forall i \in \lceil m \rfloor, ~~ \rho(\lceil m \rfloor)=mn$ ~ and ~ $\rho(\{m+1\})=l$. \newline
(C2) For every receiver $R_{i}=(\mathcal{W}_{i},\mathcal{H}_{i}) \in \mathcal{R}$ described by $(D_i,K_i)$, rank $([AD_i ~~ AK_i ~~ A_{m+1}])= \text{ rank }([AK_i ~~ A_{m+1}])$, where $A=[A_1 ~ A_2 \ldots A_m]$.
\end{theorem}
\begin{proof}
First we prove the 'if' part. Consider a discrete polymatroid $\mathbb{D}$ of rank $mn$ representable over $\mathbb{F}_{q}$ with representation $A_1,A_2,\ldots, A_{m+1},$ satisfying conditions (C1) and (C2). The matrix $A$ is the concatenation of matrices $A_1,A_2,\ldots,A_{m}$. Condition (C1) implies that $A_{i}$ is $mn \times n$ matrix for $i \in \lceil m \rfloor$ and $A_{m+1}$ is $mn \times l$ matrix. From (C1) we have that $rank(A)=mn$ making it invertible. Define $A_{i}'=A^{-1}A_{i}, i \in \lceil m+1 \rfloor$. Consider the map $f: \mathbb{F}_{q}^{mn} \rightarrow \mathbb{F}_q^l$ given by $f(X)=XA_{m+1}'$. We show that the map $f$ forms an index code of length $l$ and dimension $n$ over $\mathbb{F}_{q}$. Consider any receiver $R_{i} = (\mathcal{W}_{i}, \mathcal{H}_{i})$ described by $(D_{i},K_{i})$. From (C2) we have that the column span of the matrix $AD_i$ belongs to the span of columns of $AK_i$ and $A_{m+1}$. Matrix $AD_{i}$ can be written as $[AK_i ~~ A_{m+1}]M_i$ where $M_i$ is an $(\vert \mathcal{H}_i \vert + l) \times \vert \mathcal{W}_i \vert$ matrix. Premultiplying by $A^{-1}$, we have $[K_{i} ~~ A_{m+1}']M_i=D_i$. Hence $XD_{i}$ can be obtained at receiver $R_{i}$ from $XK_{i}$ and $XA_{m+1}'$.

To prove the 'only if' part, we assume that a vector linear index code $f$ over $\mathbb{F}_{q}$ of length $l$ and dimension $n$ exists for the generalized index coding problem $\mathcal{I}(X, \mathcal{R})$. The vector linear index code $f$ can be written as $f(X)=XA_{m+1}$ where $A_{m+1}$ is a matrix of size $mn \times l$. Let $I$ be the identity matrix of size $mn \times mn$. For $i \in \lceil m \rfloor $, let $A_{i}$ be the matrix obtained by taking only the $i(n-1)+1^{th}$ to $in^{th}$ columns of $I$. Let $V_{i}$ be the column span of $A_{i}$. We claim that the discrete polymatroid $\mathbb{D}(V_1,V_2,\ldots,V_{m+1})$ satisfies the condition (C1) and (C2). Since the concatenation of matrices $A_{i}, i \in \lceil m \rfloor$ forms an identity matrix condition (C1) is satisfied. Consider a receiver $(D_i,K_i) \in \mathcal{R}$. Since the vector index code $XA_{m+1}$ satisfies the receiver, $XD_{i}$ can be obtained from $XK_{i}$ and $XA_{m+1}$. Since $A$ is identity matrix condition (C2) is satisfied.
\end{proof}

Theorem \ref{thm:GICDPM} is a generalization of the result obtained in \cite{LICDPM} where vector linear solution of a conventional index coding problem was connected to discrete polymatroids. The result in \cite{LICDPM} can be obtained from this result by imposing the restriction on the structure of matrices $D_i$ and $K_i$. For conventional index coding the only non zero entries of $D_i$ matrix will form an identity matrix and the non zero entries of $K_{i}$ matrix forms a collection of identity matrices each corresponding to the message known at the receiver. By imposing the restrictions, condition (C2) can be expressed in terms of the elements of the ground set. In the remaining part of this section, we illustrate Theorem \ref{thm:GICDPM} with an example. 

\begin{example}
Consider the generalized index coding problem of Example \ref{Eg:GIC1}. There are five messages and since the solution is scalar, dimension is one. Consider the set of matrices  \begin{equation*}
   A_1=\begin{bmatrix}
          1 \\
          0 \\
          0 \\
          0 \\
          0 
        \end{bmatrix},
           A_2=\begin{bmatrix}
          0 \\
          1 \\
          0 \\
          0 \\
          0 
        \end{bmatrix},
           A_3=\begin{bmatrix}
          0 \\
          0 \\
          1 \\
          0 \\
          0 
        \end{bmatrix},
           A_4=\begin{bmatrix}
          0 \\
          0 \\
          0 \\
          1 \\
          0 
        \end{bmatrix},
           A_5=\begin{bmatrix}
          0 \\
          0 \\
          0 \\
          0 \\
          1 
        \end{bmatrix}.
\end{equation*} Also let $A_6 = L$, the matrix corresponding to the index code of Example \ref{Eg:GIC1}. Let $V_i$ denote the column span of $A_i$ for $i \in \lceil 6 \rfloor$. The discrete polymatroid $\mathbb{D}(V_1,V_2,\ldots,V_6)$ satisfies the conditions (C1) and (C2) of Theorem \ref{thm:GICDPM}. Rank of the discrete polymatroid is equal to five since the vector spaces $V_1,V_2,\ldots,V_5$ are linearly independent. Rank of the vector space $V_6$ is equal to three which is the length of the index code. We illustrate condition (C2) for receiver $R_5$. The matrix \begin{equation*}
   AD_5=\begin{bmatrix}
          0 \\
          0 \\
          1 \\
          1 \\
          1 
        \end{bmatrix},
        AK_5=\begin{bmatrix}
          0 & 1 \\
          1 & 0\\
          0 & 1\\
          0 & 0\\
          0 & 0
        \end{bmatrix},A_6=\begin{bmatrix}
          1 & 0 & 0 \\
          1 & 0 & 0 \\
          0 & 1 & 0 \\
          0 & 1 & 0 \\
          0 & 0 & 1
       \end{bmatrix}.
        \end{equation*} Clearly $AD_5$ lies in the column span of the matrix $[AK_5 ~ A_6]$. Condition (C2) can be similarly verified for every receiver.
\end{example}

\section{Generalized Index Coding from Discrete Polymatroids}
\label{Sec:GICfromDPM}

Discrete polymatroids can be viewed as a generalization of matroids as explained in Section \ref{Sec:MatandDPM}. In \cite{ICMT}, an index coding problem was constructed from a matroid and relationship between multilinear representation of matroids and vector linear solution of the constructed index coding problem was obtained. This was generalized to the construction of index coding problems from discrete polymatroids in \cite{LICDPM}. In this section we show the construction of a generalized index coding problem from a discrete polymatroid. The construction is similar to the construction in \cite{LICDPM}. The difference is in the set of receivers constructed from minimal excluded vectors of the discrete polymatroid $\mathbb{D}$ and that these receivers possess linear functions as \textit{Has-set}. A generalized index coding problem $\mathcal{I}_{\mathbb{D}}(Z,\mathcal{R})$ is constructed from the discrete polymatroid $\mathbb{D}$. A connection between a perfect linear solution for the generalized index coding problem $\mathcal{I}_{\mathbb{D}}(Z,\mathcal{R})$ and the representability of the discrete polymatroid $\mathbb{D}$ is established in this section. 

Consider a discrete polymatroid $\mathbb{D}$ on the ground set $\lceil r \rfloor$ with rank function $\rho$ and $\rho( \lceil r \rfloor) = k $. The generalized index coding problem $\mathcal{I}_{\mathbb{D}}(Z,\mathcal{R})$ is given below.
\begin{itemize}
%\item[(i)] The set of source messages $Z=X \cup Y$, where $X=\{x_1,x_2,\ldots,x_k\}$ and $Y=\{y_1,y_2,\ldots,y_1^{\rho\{1\}},\ldots,y_r^1,y_r^2,\ldots,y_r^{\rho\{r\}}\}$.
\item[(i)] The set of source messages $Z=X \cup Y$, where $X=\{x_1,x_2,\ldots,x_k\}$ and $Y=\{y_1^1,y_1^2,\ldots,y_1^{\rho\{1\}},\ldots,y_r^1,y_r^2,\ldots,y_r^{\rho\{r\}}\}$.

\item[(ii)] The set of receivers $\mathcal{R}$ is a union of three types of receivers $R_1,R_2$ and $R_3$ defined below. Let $\zeta_{i}= \{y_{i}^{1},y_{i}^{2},\ldots,y_{i}^{\rho(\{i \})} \}$. \newline

Receivers in $R_1$ : For a basis vector $b=\sum_{i \in \lceil r \rfloor}b_i \epsilon_{i,r} \in \mathcal{B}(\mathbb{D})$, we define the set $S_1(b)=\{(x_j,\underset{l \in (b)_{>0}}{\cup} \eta_{l}) : j \in \lceil k \rfloor, \eta_l \subseteq \zeta_l. \vert \eta_l \vert = b_l \} $. $R_1  = \underset{b \in \mathcal{B}(\mathbb{D})}{\cup}S_{1}(b)$ is the union of all such receivers for every basis of the discrete polymatroid $\mathbb{D}$.  \newline

Receivers in $R_2$ : For a minimal excluded vector $c=\sum_{i \in \lceil r \rfloor} c_i \epsilon_{i,r} \in \mathcal{C}(\mathbb{D}), j \in (c)_{>0}$ and $p \in \lceil \rho(\{j \}) \rfloor$, define the set $S_{2}(c,j,p)$ as below.  
\begin{equation*}
\begin{split}
S_2(c,j,p)= \{ (y_j^p, \underset{y \in \Gamma_1 \cup \Gamma_2}{\sum y} ): \Gamma_1=\underset{l \in (c)_{>0} \setminus\{j\}}{\cup} \eta_l, \\ \eta_l \subseteq \zeta_l, \vert \eta_l \vert = c_l, \Gamma_2 \subseteq \zeta_j \setminus \{y_j^p \}, \vert \Gamma_2 \vert = c_j - 1 \}.
\end{split}
\end{equation*}
Define $R_2 = \underset{c \in C(\mathbb{D})}{\cup} ~~ \underset{j \in (c)_{>0}}{\cup} ~~ \underset{p \in \lceil \rho(\{j\})\rfloor}{\cup} S_2(c,j,p) .$

Receivers in $R_3$ : Define $R_3= \{ (y_i^j, X) : i \in \lceil r \rfloor, j \in \lceil \rho(\{i\}) \rfloor$ \}.
\end{itemize}

Note that the minimum number of transmissions required by the above problem is $n \sum_{i \in \lceil r \rfloor} \rho ( \{i \}).$ This can be seen from the receivers in the set $R_3$. We connect the problem of representation of the discrete polymatroid $\mathbb{D}$ to the existence of a linear index coding solution of certain length for the constructed index coding problem $\mathcal{I}_\mathbb{D}(Z, \mathcal{R})$ in Theorem \ref{thm:GICfromDPM}.

\begin{theorem}
\label{thm:GICfromDPM}
If a perfect linear index coding solution of dimension $n$ over $\mathbb{F}_2$ exists for the generalized index coding problem $\mathcal{I}_\mathbb{D}(Z, \mathcal{R})$, then the discrete polymatroid $n\mathbb{D}$ is representable over $\mathbb{F}_2$.
\end{theorem}
\begin{proof}
Let $t=k+ \sum_{i=1}^r \rho(\{i \})$ denote the number of messages in the index coding problem $\mathcal{I}_\mathbb{D}(Z,\mathcal{R})$. If a perfect linear index coding solution of dimension $n$ over $\mathbb{F}_{q}$ exists for the index coding problem $\mathcal{I}_\mathbb{D}(Z, \mathcal{R})$, then from Theorem \ref{thm:GICDPM}, there exists a discrete polymatroid $\mathbb{D}'$ representable over $\mathbb{F}_{q}$ satisfying conditions (C1) and (C2). Discrete polymatroid $\mathbb{D}'$ has rank $nt$ and is over the ground set $\lceil t+1 \rfloor$. Let $V_1,V_2, \ldots,V_{t+1}$ be the vector spaces over $\mathbb{F}_q$ which forms the representation of $\mathbb{D}'$. The vector spaces $V_{i}, i \in \lceil t \rfloor$ can be expressed as the column span of matrices $A_i$ of order $nt \times n$. The vector space $V_{t+1}$ can be written as the column span of $A_{t+1}$ of order $nt \times n\sum_{i=1}^r \rho(\{i\})$. The matrix $B=[A_1,A_2,\ldots,A_t]$ is invertible from (C1). We can assume it to be identity without loss of generality. Otherwise, define $A_i'=B^{-1} A_i, i \in \lceil t+1 \rfloor$ and vector spaces given by column spans of $A_i'$ will form a representation of $\mathbb{D}'$.

The matrix $A_{t+1}$ is a $nt \times  n\sum_{i=1}^r \rho(\{i\})$ matrix and we can also assume the matrix to have a specific structure. This is because the presence of receivers belonging to $R_3$. Let $A_{t+1}=[C^T D^T]^T$ where $C$ is of order $nk \times  n\sum_{i=1}^r \rho(\{i\})$ and $D$ is of the order $ n\sum_{i=1}^r \rho(\{i\}) \times  n\sum_{i=1}^r \rho(\{i\})$. The matrix $D$ has to be full rank because of the presence of receivers in $R_3$. We can assume $D$ to be identity because if not we can define $A_{t+1}'=A_tD^{-1 }$ and it still continues to be a valid representation. Let $C_i, i \in \lceil r \rfloor$, denote the matrix obtained by taking only the $(n \sum_{j=1}^{i-1} \rho( \{j\} ) + 1 )^{\text{th}}$ to $(n \sum_{j=1}^{i} \rho({i}))^{\text{th}}$ columns of $C$. Let $C_{i,j}, j \in \lceil \rho(\{i\}) \rfloor$ denote the $nk \times n$ matrix obtained by taking the $(j-1)n + 1^{\text{th}}$ to $jn^{\text{th}}$ columns of $C_i$. Let $V_{i}'$ denote the column span of $C_{i}$ and $V'_{i,j}$ denote the column span of $C_{i,j}$. We show that the vector subspaces $V_{i}', i \in \lceil r \rfloor$ forms a representation for the discrete polymatroid $n \mathbb{D}$.

Consider a set $S \subseteq \lceil r \rfloor$. Let $b = \text{arg} \underset{b \in \mathbb{D}}{\text {max}} \vert b(S) \vert$. Let $b_i^S$ denote the $i^{th}$ component of $b^S$. Choose $b_i^S$ vector subspaces from the set $ \mathcal{V}_i = \{V_{i,j}' : j \in \lceil \rho(\{i \}) \rfloor \}$, denoted as $V'_{i,o_1},V'_{i,o_2}, \ldots, V'_{i,o_{b_i^S}}$ for every $i \in \lceil r \rfloor$. Let $\widehat{V_i} = \sum_{j \in \lceil b_i^S \rfloor} V'_{i,o_j}$ and $\widehat{C_i}= \sum_{j \in \lceil b_i^S \rfloor}C_{i,o_j}$. $\widehat{C_i}$ is a column vector which is the sum of $\vert b_i^S \vert$ column vectors of matrix $C_{i}$. From the fact that (C2) needs to be satisfied for the receivers belonging to $S_1(b_S)$, we have $dim(\sum_{i \in \lceil r \rfloor} \widehat{V_i})=n ~ rank(\mathbb{D})$. This implies that $dim( \sum_{i \in S} \widehat{V_i} ) = n \vert b^S(S) \vert$. Since the vector space $\widehat{V_i}$ is a subspace of $V_{i}'$, we have $dim(\sum_{i \in S}V_i') \leq n \rho(S).$

Let $S=\{s_1,s_2,\ldots,s_m \} \cup \{s_{m+1},s_{m+2}, \ldots s_l \}$, where $b_{s_i}^S < \rho( \{ s_i \})$ for $ i \in \lceil m \rfloor $ and $b_{s_{i}}^S = \rho(\{ s_i \})$ for $i \in \{m+1, m+2, \ldots, l \} $. Consider the vector $u = (b_{s_1}^S + 1)\epsilon_{s_1,r} + \sum_{i \in S \setminus \{s_1 \}} b_i^S \epsilon_{i,r}$. Since the vector does not belong to the discrete polymatroid it is an excluded vector. This implies that there exists a minimum excluded vector $u_m$ for which $u_m \leq u $. The $s_1^{\text{th}}$ component of $u_m$ has to be $b_{s_{1}}^S+1$. The vector $u_{m}$  can be written as $(b_{s_1}^S+1) \epsilon_{s_1,r} + \sum_{i \in S \setminus s_1} c_i^S \epsilon_{i,r},$ where $c_i^S \leq b_i^S$. Consider the receivers belonging to the set $S_2(u_m,s_1,p),$ where $ p \in \lceil \rho( \{s_1 \}) \rfloor \setminus \{o_1,o_2, \ldots, o_{b_{s_{1}}^S} \}$, it follows that 
\[C_{s_1,p}= \big( \underset{i \in (u_m)_{>0} \setminus \{s_1 \}}{\sum} \widehat{C_i}  ~ \big) + \widehat{C_{s_1}}. 
\]
%
%\[ \underset{p \in \lceil \rho( \{S_1 \}) \rfloor \setminus \{o_1,o_2, \ldots, o_{b_{s_{1}}^S} \}}{\sum} V_{s_1,p}' \subseteq \underset{i \in (u_m)_{>0} \setminus \{s_1 \}}{\sum} \widehat{V_i} + \underset{j \in \lceil b_{s_1}^S \rfloor}{\sum} V'_{s_1,o_j}.
%\]
This is true for every $ p \in \lceil \rho( \{s_1 \}) \rfloor \setminus \{o_1,o_2, \ldots, o_{b_{s_{1}}^S} \}$. Note that the vector space $V_{s_1,p}$ is the column span of matrix $C_{s_{1},p}$. It is true for any $\vert b_{s_1}^S \vert$ columns chosen in $\widehat{C_{s_1}}$. It follows that the vector space $V'_{s_1,p}$ is a subspace of $\sum_{i \in (u_m)_>0} \widehat{V_i} $ for all $p \in \rho(\{s_1\})$.
 From this, we obtain that $\sum_{p \in \lceil  \rho( \{ s_1 \})\rfloor} V'_{s_1,p} \subseteq \sum_{i \in (u_m)_>0} \widehat{V_i} \subseteq \sum_{i \in S } \widehat{V_i}$. By a similar reasoning, $V_{s_j}' \subseteq \sum_{i \in S} \widehat{V_i}, \forall j \in \lceil m \rfloor$. Since $b_{s_j}^S = \rho ( \{s_j \})$, for $j \in \{ m +1,m+2, \ldots,l \}$, we have $V_{s_{j}}'=\widehat{V_{s_{j}}}$ for $j \in \{ m+1, m+2, \ldots, l \}$. From the above facts we have $\sum_{i \in S} V_i' \subseteq \sum_{i \in S} \widehat{V_i}$. Hence, $dim(\sum_{i \in S} V_i') \leq dim( \sum_{i \in S} \widehat{V_i}) = n \rho(S)$. Thus we have established that $dim( \sum_{i \in S} V_{i}')=n \rho(S)$ for an arbitrary subset $S \subseteq \lceil r \rfloor$.
\end{proof}

In Theorem \ref{thm:GICfromDPM}, a generalized index coding problem is constructed from a discrete polymatroid and then it is shown that the discrete polymatroid is representable over the field $\mathbb{F}_2$ if a perfect linear index coding solution exists for the constructed generalized index coding problem. We illustrate the theorem in  Example \ref{Eg:TheoremGICfromDPM}. The converse of this result is however not true. In Example \ref{Eg:ConverseGICDPM}, from a binary representable discrete polymatroid we construct a generalized index coding problem for which there is no perfect linear index coding solution. 

%In Theorem \ref{thm:GICfromDPM}, a generalized index coding problem is constructed from a discrete polymatroid and then the discrete polymatroid is representable over binary if a perfect linear index coding solution exists. We illustrate the theorem in  Example \ref{Eg:TheoremGICfromDPM}. The converse of this result is however not true. In Example \ref{Eg:ConverseGICDPM}, from a binary representable discrete polymatroid we construct a generalized index coding problem for which there is no perfect linear index coding solution. 

\begin{example}
\label{Eg:TheoremGICfromDPM}
Consider the discrete polymatroid $\mathbb{D}$ on the ground set $\lceil 3 \rfloor$ with the rank function $\rho$ given by $\rho\{1\}=\rho\{2\}=1, \rho\{1,2\}=\rho\{3\}=2$ and $\rho\{1,3\}=\rho\{2,3\}=\rho\{1,2,3\}=3$. From the discrete polymatroid $\mathbb{D}$ we construct the generalized index coding problem $\mathcal{I}_\mathbb{D}(Z,\mathcal{R})$.

The set of messages possessed by source is $Z=\{x_1,x_2,x_3\} \cup \{y_1^1,y_2^1,y_3^1,y_3^2\}$. The set of receivers are constructed as in the theorem above. The set of basis vectors of the discrete polymatroid $\mathbb{D}$ is $\mathcal{B}(\mathbb{D})=\{(1,1,1),(1,0,2),(0,1,2)\}$.
We have, 
\begin{equation*}
\begin{split}
S_1((1,1,1)) & =\{(x_i,\{y_1^1,y_2^1,y_3^j\}): i \in \lceil 3 \rfloor, j \in \lceil 2 \rfloor \},\\ 
S_1 ((1,0,2)) & = \{ (x_i,\{y_1^1,y_3^1,y_3^2 \}) : i \in \lceil 3 \rfloor \} ,\\
S_1((0,1,2)) & = \{(x_i, \{y_2^1,y_3^1,y_3^2 \}) : i \in \lceil 3 \rfloor \}, 
\text{ and }\\
R_1= S_1((1, &1,1)) \cup S_1((1,0,2)) \cup S_1((0,1,2)).
\end{split}
\end{equation*}
There is only one excluded vector $(1,1,2)$. We have,
\begin{equation*}
\begin{split}
S_2((1,1,2),1,1)&=\{ (y_1^1, \{y_2^1+y_3^1+y_3^2 \}) \}, \\
S_2((1,1,2),2,1)&=\{(y_2^1, \{y_1^1+y_3^1+y_3^2 \} ) \}, \\
S_2((1,1,2),3,1)&=\{(y_3^1, \{y_1^1+y_2^1+y_3^2 \} ) \}, \\
S_2((1,1,2),3,2)&=\{(y_3^2, \{y_1^1+y_2^1 + y_3^2 \} ) \}, \text{ and } \\
R_2 & =   ~ \underset{j \in (c)_{>0}}{\bigcup} ~ \underset{p \in \lceil \rho(\{j\}) \rfloor}{\bigcup} ~ S_2((1,1,0),j,p).
\end{split}
\end{equation*}

Third set of receivers $R_3$ is a collection of four receivers $(y_1^1, X),(y_2^1, X),(y_3^1, X),(y_3^2, X)$ where $X=\{x_1,x_2,x_3\}$.
 
Note that $\mu(\mathcal{I}_\mathbb{D}(Z.\mathcal{R}))=4.$ Consider the perfect index code in which the source transmits $y_1+x_1,y_2+x_2,y_3^1+x_3$ and $y_3^2+x_1+x_2+x_3$. It can be verified that the index code satisfies the demands of all the receivers and by Theorem \ref{thm:GICfromDPM} that the discrete polymatroid $\mathbb{D}$ has a representation given by the representing matrix 
\[
 A=\begin{bmatrix}
			~~1 & ~~0 & 0 & 1 \\
			~~0 & ~~1 & 0 & 1 \\
			\coolunder{A_1}{~~0} & \coolunder{A_2}{~~0} & \coolunder{A_3}{1 & 1} \\
\end{bmatrix}.
\]
\vspace{5pt}
\end{example}

\begin{example}
\label{Eg:ConverseGICDPM}
Consider the discrete polymatroid $\mathbb{D}$ on the ground set $\lceil 3 \rfloor$ with the rank function $\rho$ given by $\rho \{1\}=\rho \{2\}=\rho \{2,3\}=2, \rho \{3\}=1$ and $\rho \{1,2\} = \rho \{1,3\}= \rho \{1,2,3 \} = 3.$ The generalized index coding problem $\mathcal{I}_{\mathbb{D}}(Z,\mathcal{R})$ constructed from the discrete polymatroid is given below. The set of messages $Z=\{x_1,x_2,x_3\} \cup \{y_1^1,y_1^2,y_2^1,y_2^2,y_3^1\}$. There are three types of receivers $R_1,R_2$ and $R_3$ which is given below.

The set of basis vectors for the problem is $\mathcal{B}(\mathbb{D})=\{(1,1,1),(1,2,0),(2,0,1),(2,1,0) \}$. We have,
\begin{equation*}
\begin{split}
S_1((1,1,1)) & =\{(x_i,\{y_1^j,y_2^k,y_3^1\}): i \in \lceil 3 \rfloor, j,k \in \lceil 2 \rfloor \},\\ 
S_1 ((1,2,0)) & = \{ (x_i,\{y_1^j,y_2^1,y_2^2 \}) : i \in \lceil 3 \rfloor, j \in \lceil 2 \rfloor \} ,\\
S_1((2,0,1)) & = \{(x_i, \{y_1^1,y_1^2,y_3^1 \}) : i \in \lceil 3 \rfloor \}, \\
S_1((2,1,0)) &= \{(x_{i},\{y_1^1,y_1^2,y_2^j \}) : i \in \lceil 3 \rfloor, j \in \lceil 2 \rfloor \} \text{ and }\\
R_1= S_1((1, &1,1)) \cup S_1((1,2,0)) \cup S_1((2,0,1)) \cup S_1((2,1,0)).
\end{split}
\end{equation*}

The set of minimal excluded vectors of $\mathbb{D}$ are $c_1= (0,2,1),c_2=(2,1,1)$ and $c_3=(2,2,0)$. We have, 
\begin{equation*}
\begin{split}
S_2(c_1,2,1)&=\{ (y_2^1, \{y_2^2+y_3^1 \}) \}, \\
S_2(c_1,2,2)&=\{(y_2^2, \{y_2^1+y_3^1 \} ) \}, \\
S_2(c_1,3,1)&=\{(y_3^1, \{y_2^1+y_2^2 \} ) \}, \\
S_2(c_2,1,1)&=\{(y_1^1, \{y_1^2+y_2^i + y_3^1 \} ) : i \in \lceil 2 \rfloor \}, \\
S_2(c_2,1,2)&=\{(y_1^2, \{y_1^1+y_2^i + y_3^1 \} ) : i \in \lceil 2 \rfloor \}, \\
S_2(c_2,2,1)&=\{(y_2^1, \{y_1^1+y_1^2 + y_3^1 \} ) \}, \\
S_2(c_2,2,2)&=\{(y_2^2, \{y_1^1+y_1^2 + y_3^1 \} ) \}, \\
S_2(c_2,3,1)&=\{(y_3^1, \{y_1^1+y_1^2 + y_2^i \} ) : i \in \lceil 2 \rfloor \}, \\
S_2(c_3,1,1)&=\{(y_1^1, \{y_1^2+y_2^1 + y_2^2 \} ) \}, \\
S_2(c_3,1,2)&=\{(y_1^2, \{y_1^1+y_2^1 + y_2^2 \} ) \}, \\
S_2(c_3,2,1)&=\{(y_2^1, \{y_1^1+y_1^2 + y_2^2 \} ) \},  \\
S_2(c_3,2,2)&=\{(y_2^2, \{y_1^1+y_1^2 + y_2^1 \} ) \} \text{ and } \\
R_2 = \underset{c \in \{c_1,c_2,c_3\}}{\bigcup} & ~ \underset{j \in (c)_{>0}}{\bigcup} ~ \underset{p \in \lceil \rho(\{j\}) \rfloor}{\bigcup} ~ S_2(c,j,p).
\end{split}
\end{equation*}

Third set of receivers $R_3$ is a collection of five receivers $(y_1^1, X),(y_1^2, X),(y_2^1, X),(y_2^2, X),(y_3^1, X)$ where $X=\{x_1,x_2,x_3\}$.

The discrete polymatroid $\mathbb{D}$ has a binary representation given by the representing matrix 
\begin{equation*}
 A=\begin{bmatrix}
          1 & 0 & ~0 & 1 ~ & ~~0 ~~~ \\
          0 & 1 & ~0 & 1 ~ & ~~0 ~~~\\
          \coolunder{A_1}{0 & 0} & \coolunder{A_2}{~1 & 1~} & \coolunder{A_3}{~1} ~~\\
\end{bmatrix}.
\end{equation*}
\vspace{5pt}

Though the discrete polymatroid has a binary representation, the generalized index coding problem $\mathcal{I}_{\mathbb{D}}(Z, \mathcal{R})$ constructed from it does not have a perfect binary solution. Suppose there exists a scalar perfect linear solution over $\mathbb{F}_{2}$. From Theorem \ref{thm:GICfromDPM}, a scalar perfect linear solution exists only if $\mathbb{D}$ is representable over $\mathbb{F}_{2}$. Every scalar perfect linear solution for $\mathcal{I}_{\mathbb{D}}(Z, \mathcal{R})$ can be written as $f(Z)=[y_1^1 ~ y_1^2 ~ y_2^1 ~ y_2^2 ~ y_3^1]A+[x_1 ~ x_2 ~ x_3]G$ where $A$ is a $5 \times 5$ matrix over $\mathbb{F}_{2}$ and $G$ is a $3 \times 5 $ matrix over $\mathbb{F}_{2}$. The matrix $A$ needs to be full rank to ensure that the receivers belonging to $R_3$ are satisfied which allows us to assume $A$ to be identity matrix. It can be shown by checking all possible solutions that there does not exist a $3 \times 5$ matrix $G$ over $\mathbb{F}_2$ which solves the generalized index coding problem. We provide an alternate proof here. Let $G_{i}$ denote the $i^{th}$ column of $G$. The first three columns of the matrix $G$ can be assumed to be the columns of $3 \times 3$ identity matrix. The column $G_5$ has to be $[1 ~ 1 ~ 1]^T$. The column $G_5$ cannot be $[1 ~ 0 ~ 0]^T, [ 0 ~ 1 ~ 0]^T$ and $[1 ~ 1 ~ 0]^T$, since $dim(V_1+V_3)=3$. If $G_5 = [0 ~ 1 ~ 1]^T$, the receivers $(x_i, \{ y_1^2,y_2^1,y_3^1 \}), i \in \lceil 3 \rfloor $  fails to decode the demands. Similarly if $G_5=[0 ~ 1 ~ 1]^T$ receivers $(x_i, \{y_1^1,y_2^1,y_3^1 \}), i \in \lceil 3 \rfloor$ fails and if $G_5=[1 ~ 0 ~ 1]^T$ the receivers $(x_i, \{ y_1^1,y_2^1,y_3^1 \}) , i \in \lceil 3 \rfloor$ fails to decode the demands. From the restrictions $dim(V_2+V_3)=2$ and $dim(V_2)=2$, the column vector $G_4$ has only two possibilities $[1 ~ 1 ~0]^T$ and $[1 ~ 1~1]^T$. If $G_4$ is equal to $[1 ~ 1 ~ 0]^T$ then receivers $(x_i, \{y_1^1, y_1^2, y_2^2 \})$ fails and if $G_4 = [1 ~ 1 ~ 1]^T$ then receivers $(x_i, \{ y_1^1.y_2^2,y_3^1 \})$ fails to decode the demands. This shows that there does not exist a perfect scalar linear solution over $\mathbb{F}_{2}$ for $\mathcal{I}_\mathbb{D}(Z, \mathcal{R})$.
\end{example}

\section{Matroids and Generalized Index Coding Problem}
\label{Sec:MICCSI}

In this section we construct a generalized index coding problem from a matroid. The construction explained in Section \ref{Sec:GICfromDPM} is more general than this since discrete polymatroids can be viewed as a generalization of matroids. However for the generalized index coding problem constructed from matroids, there is an if and only if relationship between the constructed index coding problem and the represent ability of the matroids as shown in Theorem \ref{th:GICfromMat}. The index code constructed from the matroid is similar to the construction provided in \cite{ICMT}. Receivers belonging to the set $R_2$, which are constructed from the circuits of matroid are different as explained below. 

\begin{definition} \label{def:matindex}
Given a matroid $\mathcal{\mathcal{M}}(Y,r)$ of rank $k$ over ground set
$Y=\{y_1,\dots, y_m\}$, we define a corresponding index coding with coded side information problem $\mathcal{I}_\mathcal{\mathcal{M}}(Z,\mathcal{R})$ as follows:
\begin{enumerate}
  \item $Z=Y\cup X$, where $X= \{x_1,\dots,x_k\}$,
   \item $\mathcal{R}=R_1\cup R_2\cup R_3$ where
  \begin{enumerate}
  \item  $R_1=\{ (x_i,B);  B\in \mathcal{B}(\mathcal{M}), i=1,\dots,k\}$
  \item  $R_2=\{ (y, \underset{y_j \in C\setminus\{y\}}{\sum} y_j );C\in \mathfrak{C}(\mathcal{M}), y\in C\}$
  \item $R_3=\{(y_i,X); i=1,\dots,m\}$
  \end{enumerate}
\end{enumerate}
\end{definition}

\begin{theorem}\label{th:GICfromMat}
Consider a matroid $\mathcal{M}(Y,r)$ on the ground set \mbox{$Y=\{y_1,\dots,y_m\}$}, and $\mathcal{I}_{\mathcal{M}}(Z,\mathcal{R})$ be the corresponding generalized index coding problem constructed from it. Then, the matroid $\mathcal{M}$ has a linear representation over $\mathbb{F}_2$ if and only if there exists a perfect scalar linear index code for $\mathcal{I}_{\mathcal{M}}(Z, \mathcal{R})$ over $\mathbb{F}_2$.
\end{theorem}

\begin{proof}
 Let $\xi=(x_{1},\ldots,x_{k})\in \mathbb{F}_2^{k}$, and $\chi=(y_1,\ldots,y_{m},x_{1},\dots,x_{k})\in \mathbb{F}_2^{(m+k)}.$

We first assume that the matroid $\mathcal{M}$ is representable and show the existence of a perfect scalar linear index code for the index coding problem $\mathcal{I}_{\mathcal{M}}(Z, \mathcal{R})$. Let $M$ be the matrix representing the matroid $\mathcal{M}$. Since the matroid $\mathcal{M}$ is of rank $k$, matrix $M$ is a $k \times m $ matrix.

Consider the following linear map $f(\chi)=(f_1(\chi),\dots, f_m(\chi))$ where
$$f_i(\chi)=y_i+\xi M_i \in \mathbb{F}_2, i=1,\dots,m.$$ Note that $f$ is a map from $\mathbb{F}_{2}^{m+k}$ to $\mathbb{F}_{2}^m$. We show that $f$ is a perfect scalar linear index code for $\mathcal{I}_{\mathcal{M}}(Z, \mathcal{R})$. To show this we show that all the receivers are able to satisfy their demands using their \textit{Has-sets} and the transmitted messages. 
\begin{itemize}
  \item Receiver $R_{1}$ : Consider a  basis $B=\{y_{i_1},\dots,y_{i_k}\}\in\mathcal{B}(\mathcal{M})$, and let $\rho_i=(x_i,B)\in R_1$,  $i=1,\dots,k$. We have $f_{i_{j}}(\chi)=y_{i_{j}}+\xi M_{i_{j}} , j=1,2,\ldots,k$. Combining these equations we obtain 
\begin{flalign*}
\begin{split}
[f_{i_{1}}(\chi) ~ f_{i_{2}}(\chi) \ldots ~ f_{i_{k}}(\chi)]=&[y_{i_{1}} ~ y_{i_{2}} \ldots ~ y_{i_{k}}] + \\ & \xi[M_{i_{1}} M_{i_{2}} \ldots M_{i_{k}}].
\end{split}
\end{flalign*}  
Since $\{y_{i_1},\dots,y_{i_k}\}\in\mathcal{B}(\mathcal{M})$ the matrix formed by concatenation of $M_{i_{1}}, M_{i_{2}}, \ldots , M_{i_{k}}$ is invertible. Let $B=[M_{i_{1}} ~ M_{i_{2}} \ldots ~ M_{i_{k}}]$. The receivers can obtain $\xi$ using the relation that $$\xi = [f_{i_{1}}(\chi)-y_{i_{1}} ~ f_{i_{2}}(\chi)-y_{i_{2}} \ldots f_{i_{k}}(\chi)-y_{i_{k}}]B^{-1}.$$

  \item Receiver $R_{2}$ : Let $C=\{y_{i_1},\dots,y_{i_c}\}\in \mathfrak{C}(\mathcal{M})$ and $\rho=(y_{i_1},\underset{y_j \in C\setminus\{y_{i_{1}}\}}{\sum} y_j )\in R_2$. Let $C'=C \setminus y_{i_1}$. We have $f_{i_{j}}(\chi)=y_{i_{j}}+\xi M_{i_{j}} , j=1,2,\ldots,c$. From this we can establish the relation 
\begin{flalign*}
\begin{split}
f_{i_{2}}(\chi) + \ldots + f_{i_{c}}(\chi)=&y_{i_{2}}+\ldots + y_{i_{c}} + \\ & \xi(M_{i_{2}}+\ldots+M_{i_{c}}).
\end{split}
\end{flalign*}  
Since the matroid is representable over a binary field we have $M_{i_{1}}=M_{i_{2}}+M_{i_3}+\ldots + M_{i_{c}}$. Receiver can decode its demanded message $y_{i_{1}}$ using the relation 
\begin{flalign*}
\begin{split}
y_{i_{1}} = (f_{i_{1}}(\chi)+ f_{i_{2}}(\chi) & +\ldots+f_{i_{c}}(\chi)) + \\ & (y_{i_{2}}+ \ldots+y_{i_{c}}).
\end{split}
\end{flalign*}
In a similar way all receivers belonging to $R_{2}$ can decode their demanded messages.
 \item Receiver $ R_{3}$ : For all $\rho=(y_i,X)\in R_3$, receivers can obtain its demanded message using the relation $y_{i}=f_{i}(\chi)-\xi M_i$.
\end{itemize}
The index code is clearly linear and also $\mu(\mathcal{I}_{\mathcal{M}}(Z,\mathcal{R}))=m$. Hence the code defined by the map $f$ is a perfect linear index code. 

Now, suppose that there exists a perfect scalar linear index code for $\mathcal{I}_{\mathcal{M}}(Z,\mathcal{R})$. We have to show that this will induce a linear representation of the matroid $\mathcal{M}$  over $\mathbb{F}_2$.

Let $g: (\mathbb{F}_2)^{m+k} \longrightarrow (\mathbb{F}_2)^m$ be a perfect scalar linear index code for $\mathcal{I}_{\mathcal{M}}(Z,\mathcal{R})$ over the field $\mathbb{F}_2$. The index code $g$ can be written as 
\[ g( \chi)=[x_1 ~ x_2 \ldots x_k] A+ [y_1 ~ y_2 \ldots y_m]B
\] where $A$ is a $k \times m$ matrix and $B$ is a $m \times
m$ matrix over $\mathbb{F}_{2}$.

%Let $g(\chi)=(g_1(\chi),\dots,g_m(\chi)),$ where  $g_i(\chi)\in \mathbb{F}_2.$ The map $g_i$ can be written as 
% $$g_i(\chi)=\sum_{j=1}^{k}x_jA_{ij}+\sum_{j=1}^m y_jB_{ij},$$
%where $A_{ij}, B_{ij}\in \mathbb{F}_2$.

Since all the receivers belonging $R_{3}$ are satisfied by the index code the matrix $B$ is invertible. Consider the function $f:
(\mathbb{F}_2)^{m+k} \longrightarrow (\mathbb{F}_2)^m$, given by $f(\chi)=g(\chi)B^{-1},
\forall \chi \in (\mathbb{F}_2)^{m+k} $. Note that  $f$ is a valid index code for $\mathcal{I}_{\mathcal{M}}(Z,\mathcal{R})$. The function $f$ can be written as 
\[
f(\chi)=[y_1 ~ y_2 \ldots y_m]+ [x_1 x_2 \ldots x_k]C,
\]
where $C$ is a $k \times m$ matrix over $\mathbb{F}_2$. We show that the matrix $C$ represents the matroid $\mathcal{M}$. We need to show that rank$(C_X)= r(X)$ for all $X \subseteq \lceil m \rfloor$. It suffices to show for all subsets which forms the bases and circuits of matroid $\mathcal{M}$.

Let $B\in \mathcal{B}(\mathcal{M})$ a basis. Then the receivers $(x_j,B), j=1,\dots,k$, belonging to $R_1$ will be able to decode their required messages if and only if $C_B$ is invertible. Therefore, $\rank(C_B)=k= r(B).$

Consider a circuit  $S\in \mathfrak{C}(\mathcal{M})$. Consider a receiver $ \rho \in R_2$ which demands a message corresponding to an element in circuit and possess the sum of messages corresponding to the remaining elements in circuit. Let the receiver $\rho=(y_{i_1},\underset{y_j \in S \setminus\{y_{i_1}\}}{\sum} y_j )\in R_2$. The existence of the binary linear decoding function for the receiver implies that $C_{i_1}= \underset{i \in S \setminus \{y_{i_{1}}\}}{\sum} C_{i}.$ Consider the set $S'=S \setminus y_{i_1}$. The set $S'$ is an independent set of matroid $\mathcal{M}$ since $S$ is a circuit. The set $S'$ can be extended to a basis from which it follows that $rank(C_{S'})= \vert S' \vert$. We also have rank$(C_S)=$ rank$(C_{S'})$ from which we obtain that rank$(C_S)= \vert S \vert - 1= r(S)$. Since the circuit and the demanded element of the circuit was chosen arbitrarily it completes the proof. 
\end{proof}

Theorem \ref{th:GICfromMat} shows the existence of a relationship between binary representability of matroids and the solution to certain index coding problems. We use this to show that not every generalized index coding problem has a binary solution in Example \ref{Eg:U24}. Other examples are also provided which illustrates the theorem.

\begin{example}
\label{Eg:U23}
 The uniform matroid $U_{2,3}$ is defined on a ground set $Y=\{y_1,y_2,y_3\}$ of three elements, such that $\forall I\subseteq Y$ and $|I|\leq 2,r(I)=|I|$, and $r(Y)=2$. Consider a binary linear representation  of $U_{2,3}$ :
% \begin{equation*}
%   M_1=\begin{bmatrix}
%          1 & 0 & 1\\
%          0 & 1 & 1\\
%      \end{bmatrix},
%    M_2=\begin{bmatrix}
%          0 \\ 
%          1 \\
%       \end{bmatrix},
%        M_3=\begin{bmatrix}
%          1 \\
%          1 \\
%       \end{bmatrix}.
% \end{equation*}
  \begin{equation*}
   M=\begin{bmatrix}
          1 & 0 & 1\\
          0 & 1 & 1\\
      \end{bmatrix}.
\end{equation*}
 The index coding with coded side information problem corresponding to this matroid has the source messages set $\chi = \lbrace y_{1},y_{2},y_{3},x_{1},x_{2} \rbrace$,  where each message belongs to the finite field $\mathbb{F}_{2}$. There are three sets of receivers and they are given below.

\begin{itemize}
\item Receivers in $R_{1}$ : $\lbrace x_{1}, \lbrace y_{1},y_{2} \rbrace \rbrace ,\lbrace x_{2}, \lbrace y_{1},y_{2} \rbrace \rbrace $, $\lbrace x_{1}, \lbrace y_{1},y_{3} \rbrace \rbrace , \lbrace x_{2}, \lbrace y_{1},y_{3} \rbrace \rbrace , \lbrace x_{1}, \lbrace y_{2},y_{3} \rbrace \rbrace $, $\lbrace x_{2}, \lbrace y_{2},y_{3} \rbrace \rbrace $.
\item Receivers in $R_{2}$ : $\lbrace y_{1} , \lbrace y_{2}+y_{3} \rbrace \rbrace,\lbrace y_{1} , \lbrace y_{2}+y_{3} \rbrace \rbrace$,$\lbrace y_{1} , \lbrace y_{2}+y_{3} \rbrace \rbrace$
\item Receivers in $R_{3}$ : $\lbrace y_{1}, \lbrace x_{1},x_{2} \rbrace \rbrace$, \\ $\lbrace y_{2}, \lbrace x_{1},x_{2} \rbrace \rbrace$,$\lbrace y_{3}, \lbrace x_{1},x_{2} \rbrace \rbrace$
\end{itemize}
The perfect linear index coding solution for the index coding problem is given by the map $f : \mathbb{F}_{2}^{5} \rightarrow \mathbb{F}_{2}^{3}$ given by 
\[f(\chi)=[y_1 ~ y_2 ~ y_3] + [x_1 ~ x_2]M
\]. The index code is as follows. 

\begin{itemize}
\item  $c_{1}=y_{1}+x_{1}$
\item $ c_{2}=y_{2} + x_{2}$
\item $ c_{3}=y_{3} + x_{1} + x_{2}$
\end{itemize}

It can be verified that all the receivers are able to decode its demands using the transmissions and the \textit{Has-sets} available to it. Decoding procedure at receivers is given in Tables \ref{Tab:DecodeUniform1},  \ref{Tab:DecodeUniform2} and \ref{Tab:DecodeUniform3}.

\begin{table}[ht]
\centering{}
\setlength\extrarowheight{2pt}
\caption{Decoding procedures for receivers in $R_{1}$ of Example \ref{Eg:U23}.}
\begin{tabular}{|c|c|}
\hline
Receivers in $R_{1}$& Decoding Procedure \tabularnewline
\hline
$\lbrace x_{1}, \lbrace y_{1},y_{2} \rbrace \rbrace  $  & $c_{1}+y_{1}$ 
\tabularnewline
\hline
$\lbrace x_{2}, \lbrace y_{1},y_{2} \rbrace \rbrace  $ & $c_{2}+y_{2}$ \tabularnewline
\hline
$\lbrace x_{1}, \lbrace y_{1},y_{3} \rbrace \rbrace  $ & $c_{1}+y_{1}$ \tabularnewline
\hline
$\lbrace x_{2}, \lbrace y_{1},y_{3} \rbrace \rbrace  $ & $c_{3}+ c_1 + y_3 + y_1$ \tabularnewline
\hline
$\lbrace x_{1}, \lbrace y_{2},y_{3} \rbrace \rbrace  $ &  $c_{3} + c_{2} + y_{3}+y_2$ \tabularnewline
\hline
$\lbrace x_{2}, \lbrace y_{2},y_{3} \rbrace \rbrace  $ &  $c_2+y_2$ \tabularnewline
\hline
\end{tabular}
\label{Tab:DecodeUniform1}
\end{table}

\begin{table}[ht]
\centering{}
\setlength\extrarowheight{2pt}
\caption{Decoding procedures for receivers in $R_{2}$ of Example \ref{Eg:U23}.}
\begin{tabular}{|c|c|}
\hline
Receivers in $R_{2}$& Decoding Procedure \tabularnewline
\hline
$\lbrace y_{1}, \lbrace y_{2}+y_{3} \rbrace \rbrace  $  & $y_{2}+y_{3}+c_{1}+c_{2}+c_{3}$ \tabularnewline
\hline
$\lbrace y_{2}, \lbrace y_{1}+y_{3} \rbrace \rbrace  $ & $y_{1}+y_{3}+c_{1}+c_{2}+c_{3}$ \tabularnewline
\hline
$\lbrace y_{3}, \lbrace y_{1}+y_{2} \rbrace \rbrace  $ & $y_{1}+y_{2}+c_{1}+c_{2}+c_{3}$ \tabularnewline
\hline
\end{tabular}
\label{Tab:DecodeUniform2}
\end{table} 

\begin{table}[ht]
\centering{}
\setlength\extrarowheight{2pt}
\caption{Decoding procedures for receivers in $R_{3}$ of Example \ref{Eg:U23}.}
\begin{tabular}{|c|c|}
\hline
Receivers in $R_{3}$& Decoding Procedure \tabularnewline
\hline
$\lbrace y_{1}, \lbrace x_{1},x_{2} \rbrace \rbrace  $  & $x_{1}+c_{1}$ \tabularnewline
\hline
$\lbrace y_{2}, \lbrace x_{1},x_{2} \rbrace \rbrace  $ & $x_{2}+c_{2}$ \tabularnewline
\hline
$\lbrace y_{3}, \lbrace x_{1},x_{2} \rbrace \rbrace  $ & $x_{1}+x_{1}+c_{3}$ \tabularnewline
\hline
\end{tabular}
\label{Tab:DecodeUniform3}
\end{table}
\end{example}

\begin{example}
\label{Eg:U24}
Consider the following index coding problem with coded side information $\mathcal{I}(Z,\mathcal{R})$:

The set of messages $Z= \{y_1,y_2,y_3,y_4,x_1,x_2 \}$.

The set of receivers are given below. 
\begin{itemize}
\item Receivers in $R_{1}$ : 

$\lbrace (x_{i}, \lbrace y_{1},y_{2} \rbrace ), i \in \lceil 2 \rfloor \rbrace, \lbrace (x_{i}, \lbrace y_{1},y_{3} \rbrace ), i \in \lceil 2 \rfloor \rbrace$,
$\lbrace (x_{i}, \lbrace y_{1},y_{4} \rbrace ), i \in \lceil 2 \rfloor \rbrace$,
$\lbrace (x_{i}, \lbrace y_{2},y_{3} \rbrace ), i \in \lceil 2 \rfloor \rbrace$,
$\lbrace (x_{i}, \lbrace y_{2},y_{4} \rbrace ), i \in \lceil 2 \rfloor \rbrace$,
$\lbrace (x_{i}, \lbrace y_{3},y_{4} \rbrace ), i \in \lceil 2 \rfloor \rbrace$.
\item Receivers in $R_{2}$ : $\lbrace y_{1} , \lbrace y_{2}+y_{3} \rbrace \rbrace,\lbrace y_{2} , \lbrace y_{1}+y_{3} \rbrace \rbrace$,$\lbrace y_{3} , \lbrace y_{1}+y_{3} \rbrace \rbrace$,$\lbrace y_{1} , \lbrace y_{2}+y_{4} \rbrace \rbrace$,$\lbrace y_{2} , \lbrace y_{1}+y_{4} \rbrace \rbrace$, $\lbrace y_{4} , \lbrace y_{1}+y_{2} \rbrace \rbrace$,$\lbrace y_{1} , \lbrace y_{3}+y_{4} \rbrace \rbrace$,$\lbrace y_{3} , \lbrace y_{1}+y_{4} \rbrace \rbrace$,$\lbrace y_{4} , \lbrace y_{1}+y_{3} \rbrace \rbrace$,$\lbrace y_{2} , \lbrace y_{3}+y_{4} \rbrace \rbrace$,$\lbrace y_{3} , \lbrace y_{2}+y_{4} \rbrace \rbrace$ and $\lbrace y_{4}, \lbrace y_{2}+y_{3} \rbrace \rbrace$
\item Receivers in $R_{3}$ : $\lbrace (y_{i}, \lbrace x_{1},x_{2},x_{3},x_{4}), i \in \lceil 4 \rfloor \rbrace \rbrace$.
\end{itemize}

The above index coding problem is constructed from the uniform matroid $U_{2,4}$. The uniform matroid is defined on a ground set $Y=\{y_1,y_2,y_3,y_4\}$  such that $\forall I\subseteq Y$ and $|I|\leq 2,r(I)=|I|$, and $r(Y)=2$. The matroid $U_{2,4}$ does not have a binary representation. The matroid has a representation over ternary field $GF(3)$ :
 \begin{equation*}
V_1=\begin{bmatrix}
1  \\
0  \\
\end{bmatrix},
V_2=\begin{bmatrix}
0  \\
1  \\
\end{bmatrix},
V_3=\begin{bmatrix}
1  \\
1  \\
\end{bmatrix},
V_4=\begin{bmatrix}
1  \\
2  \\
\end{bmatrix}.
\end{equation*}

It can be verified that the above generalized index coding problem does not have a perfect scalar binary linear solution as implied by Theorem \ref{th:GICfromMat}. Since the matroid does not have a linear representation over binary the generalized index coding problem constructed from it does not have a scalar perfect linear solution.
\end{example}

\begin{example}
\label{Eg:Hamming}
Consider the generator matrix $G= \begin{bmatrix}
          1 & 0 & 0 & 0 & 0 & 1 & 1\\
          0 & 1 & 0 & 0 & 1 & 0 & 1\\
          0 & 0 & 1 & 0 & 1 & 1 & 0 \\
          0 & 0 & 0 & 1 & 1 & 1 & 1\\
        \end{bmatrix} $ of a $[7,4,3]$ Hamming code. The vector matroid of $G$ is a matroid having a ground set $Y=\{y_{1},y_2,\ldots,y_7 \}$. Rank of the matroid $\mathcal{M}(G)$ is four. Consider the index coding with coded side information problem $\mathcal{I}_{\mathcal{M}(G)}(Z,\mathcal{R})$  corresponding to the matroid $\mathcal{M}(G)$. 
The set of messages possessed by the source is the set $Z=\{y_1,y_2,y_3,y_4,y_5,y_6,y_7 \} \cup \{x_1,x_2,x_3,x_4 \}$. 

There are twenty eight bases to the above matroid. Each basis gives rise to four receivers in the corresponding index coding problem. The circuits of the matroid are $ \{y_1,y_2,y_4,y_7 \}$, $\{y_1,y_2,y_5,y_6 \}$, $\{y_1,y_3,y_4,y_6 \}$, $\{y_1,y_3,y_5,y_7 \}$, $\{y_2,y_3,y_4,y_5 \}$, $\{y_2,y_3,y_6,y_7 \}$ and $\{y_4,y_5,y_6,y_7\}$. Each of these circuits give rise to four receivers in the corresponding index coding problem. There are seven more receivers belonging to set of receivers $R_3=\{(y_i,\{x_1,x_2,x_3,x_4\}); i=1,\dots,7\}$. The matroid has a scalar linear representation and the corresponding index coding problem obtained from the matroid has a perfect linear solution. The length of the perfect linear index code is seven and the perfect linear index code is given by the matrix $\begin{bmatrix}
          I  \\
          G  \\
        \end{bmatrix}.$ The index code is as follows. \begin{itemize}
\item  $c_{1}=y_{1}+x_{1}$
\item $ c_{2}=y_{2} + x_{2}$
\item $ c_{3}=y_{3} + x_{3}$
\item $ c_{4}=y_{4} + x_{4}$
\item $ c_{5}=y_{5} + x_{2} + x_{3}+x_{4}$
\item $ c_{6}=y_{6} + x_{1} + x_{3} + x_{4}$
\item $c_{7}= y_{7}+x_{1}+x_{2}+x_{4}$
\end{itemize}

Details of all the receivers and the decoding procedure are given in Table \ref{Tab:DecodeHamming1},  \ref{Tab:DecodeHamming2} and \ref{Tab:DecodeHamming3}. From the tables it is clear that the above index code is a perfect scalar linear code.

\begin{table*}[h]
\centering{}
\caption{Decoding procedure for receivers in $R_{1}$ of Example \ref{Eg:Hamming}.}
\begin{tabular}{|c|c|c|c|}
\hline
Receivers & Decoding Procedure & Receivers & Decoding Procedure\tabularnewline
\hline
\multirow{ 4}{*}{ $\lbrace (x_{i}, \lbrace y_{1},y_{2},y_{3},y_{4} \rbrace ), i \in \lceil 4 \rfloor \rbrace  $ } & $c_{1}+y_{1}$ & \multirow{4}{*}{$\{(x_i,\{y_1,y_2,y_3,y_5 \}), i \in \lceil 4 \rfloor  \}$} & $c_{1}+y_{1}$\tabularnewline
& $c_{2} + y_{2}$ & & $c_{2} + y_{2}$ \tabularnewline
& $c_{3}+y_{3}$ & & $c_{3} + y_{3}$ \tabularnewline
& $c_{4}+y_{4}$ & & $c_{5}+c_2+c_3+y_2+y_3+y_5$ \tabularnewline
\hline
\multirow{ 4}{*}{ $\{(x_i,\{y_1,y_2,y_3,y_6 \}), i \in \lceil 4 \rfloor  \} $ } & $c_{1}+y_{1}$ & \multirow{4}{*}{$\{(x_i,\{y_1,y_2,y_3,y_7 \}), i \in \lceil 4 \rfloor  \} $} & $c_{1}+y_{1}$\tabularnewline
& $c_{2} + y_{2}$ & & $c_{2} + y_{2}$ \tabularnewline
& $c_{3}+y_{3}$ & & $c_{3} + y_{3}$ \tabularnewline
& $c_{6}+c_1+c_3+y_6+y_1+y_3$ & & $c_{7}+c_1+c_2+y_7+y_1+y_2$ \tabularnewline
\hline
\multirow{ 4}{*}{$\{(x_i,\{y_1,y_2,y_4,y_5 \}), i \in \lceil 4 \rfloor  \} $} & $c_{1}+y_{1}$ & \multirow{4}{*}{$\{(x_i,\{y_1,y_2,y_4,y_6 \}), i \in \lceil 4 \rfloor  \} $} & $c_{1}+y_{1}$\tabularnewline
& $c_{2} + y_{2}$ & & $c_{2} + y_{2}$ \tabularnewline
& $c_5+c_2+c_4+y_5+y_2+y_4$ & & $c_{6} + y_{6}+c_1+y_1+c_4+y_4$ \tabularnewline
& $c_{4}+y_{4}$ & & $c_4+y_4$ \tabularnewline
\hline
\multirow{ 4}{*}{$\{(x_i,\{y_1,y_2,y_5,y_7 \}), i \in \lceil 4 \rfloor  \} $} & $c_{1}+y_{1}$ & \multirow{4}{*}{$\{(x_i,\{y_1,y_2,y_6,y_7 \}), i \in \lceil 4 \rfloor  \} $} & $c_{1}+y_{1}$\tabularnewline
& $c_{2} + y_{2}$ & & $c_{2} + y_{2}$ \tabularnewline
& $c_5+y_5+c_7+y_7+c_1+y_1$ & & $c_{6} + y_{6}+c_7+y_7+c_2+y_2$ \tabularnewline
& $c_{7}+y_{7}+c_1+y_1+c_2+y_2$ & & $c_{7}+y_{7}+c_1+y_1+c_2+y_2$ \tabularnewline
\hline
\multirow{ 4}{*}{$\{(x_i,\{y_1,y_3,y_4,y_5 \}), i \in \lceil 4 \rfloor  \} $} & $c_{1}+y_{1}$ & \multirow{4}{*}{$\{(x_i,\{y_1,y_3,y_4,y_7 \}), i \in \lceil 4 \rfloor  \} $} & $c_{1}+y_{1}$\tabularnewline
& $c_{5} + y_{5}+c_3+y_3+c_4+y_4$ & & $c_{7} + y_{7}+c_1+y_1+c_4+y_4$ \tabularnewline
& $c_{3} + y_{3}$ & & $c_3+y_3$ \tabularnewline
& $c_4+y_4$ & & $c_4+y_4$ \tabularnewline
\hline
\multirow{ 4}{*}{$\{(x_i,\{y_1,y_3,y_5,y_6 \}), i \in \lceil 4 \rfloor  \} $} & $c_{1}+y_{1}$ & \multirow{4}{*}{$\{(x_i,\{y_1,y_3,y_6,y_7 \}), i \in \lceil 4 \rfloor  \} $} & $c_{1}+y_{1}$\tabularnewline
& $c_{5} + y_{5}+c_6+y_6+c_1+y_1$ & & $c_{7} + y_{7}+c_6+y_6+c_3+y_3$ \tabularnewline
& $c_{3} + y_{3}$ & & $c_3+y_3$ \tabularnewline
& $c_6+y_6+c_1+y_1+c_3+y_3$ & & $c_6+y_6+c_1+y_1+c_3+y_3$ \tabularnewline
\hline
\multirow{ 4}{*}{$\{(x_i,\{y_1,y_4,y_5,y_6 \}), i \in \lceil 4 \rfloor  \} $} & $c_{1}+y_{1}$ & \multirow{4}{*}{$\{(x_i,\{y_1,y_4,y_5,y_7 \}), i \in \lceil 4 \rfloor  \} $} & $c_{1}+y_{1}$\tabularnewline
& $c_{5} + y_{5}+c_6+y_6+c_1+y_1$ & & $c_{7} + y_{7}+c_1+y_1+c_4+y_4$ \tabularnewline
& $c_6+y_6+c_1+y_1+c_4+y_4$ & & $c_7+y_7+c_5+y_5+c_1+y_1$ \tabularnewline
& $c_4+y_4$ & & $c_4+y_4$ \tabularnewline
\hline
\multirow{ 4}{*}{$\{(x_i,\{y_1,y_4,y_6,y_7 \}), i \in \lceil 4 \rfloor  \} $} & $c_{1}+y_{1}$ & \multirow{4}{*}{$\{(x_i,\{y_1,y_5,y_6,y_7 \}), i \in \lceil 4 \rfloor  \} $} & $c_{1}+y_{1}$\tabularnewline
& $c_7 + y_7 +c_4+y_4+c_1+y_1$ & & $c_6 + y_6+c_5+y_5+c_1+y_1$ \tabularnewline
& $c_6+y_6+c_1+y_1+c_4+y_4$ & & $c_7+y_7+c_5+y_5+c_1+y_1$ \tabularnewline
& $c_4+y_4$ & & $c_7+y_7+c_6+y_6+c_5+y_5$ \tabularnewline
\hline
\multirow{ 4}{*}{$\{(x_i,\{y_2,y_3,y_4,y_6 \}), i \in \lceil 4 \rfloor  \} $} & $c_{6}+y_{6}+c_3+y_3+c_4+y_4$ & \multirow{4}{*}{$\{(x_i,\{y_2,y_3,y_4,y_7 \}), i \in \lceil 4 \rfloor  \} $} & $c_7+y_7+c_2+y_2+c_4+y_2$\tabularnewline
& $c_2+y_2$ & & $c_2+y_2$ \tabularnewline
& $c_3+y_3$ & & $c_3+y_3$ \tabularnewline
& $c_4+y_4$ & & $c_4+y_4$ \tabularnewline
\hline
\multirow{ 4}{*}{$\{(x_i,\{y_2,y_3,y_5,y_6 \}), i \in \lceil 4 \rfloor  \} $} & $c_6+y_6+c_5+y_5+c_2+y_2$ & \multirow{4}{*}{$\{(x_i,\{y_2,y_3,y_5,y_7 \}), i \in \lceil 4 \rfloor  \} $} & $c_7+y_7+c_5+y_5+c_3+y_3$\tabularnewline
& $c_2+y_2$ & & $c_2+y_2$ \tabularnewline
& $c_3+y_3$ & & $c_3+y_3$ \tabularnewline
& $c_5+y_5+c_3+y_3+c_2+y_2$ & & $c_5+y_5+c_3+y_3+c_2+y_2$ \tabularnewline
\hline
\multirow{ 4}{*}{$\{(x_i,\{y_2,y_4,y_5,y_6 \}), i \in \lceil 4 \rfloor  \} $} & $c_6+y_6+c_5+y_5+c_2+y_2$ & \multirow{4}{*}{$\{(x_i,\{y_2,y_4,y_5,y_7 \}), i \in \lceil 4 \rfloor  \} $} & $c_7+y_7+c_4+y_4+c_2+y_2$\tabularnewline
& $c_2+y_2$ & & $c_2+y_2$ \tabularnewline
& $c_5+y_5+c_4+y_4+c_2+y_2$ & & $c_5+y_5+c_4+y_4+c_2+y_2$ \tabularnewline
& $c_4+y_4$ & & $c_4+y_4$ \tabularnewline
\hline
\multirow{ 4}{*}{$\{(x_i,\{y_2,y_4,y_6,y_7 \}), i \in \lceil 4 \rfloor  \} $} & $c_7+y_7+c_4+y_4+c_2+y_2$ & \multirow{4}{*}{$\{(x_i,\{y_2,y_5,y_6,y_7 \}), i \in \lceil 4 \rfloor  \} $} & $c_6+y_6+c_5+y_5+c_2+y_2$\tabularnewline
& $c_2+y_2$ & & $c_2+y_2$ \tabularnewline
& $c_7+y_7+c_6+y_6+c_2+y_2$ & & $c_7+y_7+c_6+y_6+c_2+y_2$ \tabularnewline
& $c_4+y_4$ & & $c_7+y_7+c_6+y_6+c_5+y_5$ \tabularnewline
\hline
\multirow{ 4}{*}{$\{(x_i,\{y_3,y_4,y_5,y_6 \}), i \in \lceil 4 \rfloor  \} $} & $c_6+y_6+c_4+y_4+c_3+y_3$ & \multirow{4}{*}{$\{(x_i,\{y_3,y_4,y_5,y_7 \}), i \in \lceil 4 \rfloor  \} $} & $c_7+y_7+c_5+y_5+c_3+y_3$\tabularnewline
& $c_5+y_5+c_4+y_4+c_3+y_3$ & & $c_5+y_5+c_4+y_4+c_3+y_3$ \tabularnewline
& $c_3+y_3$ & & $c_3+y_3$ \tabularnewline
& $c_4+y_4$ & & $c_4+y_4$ \tabularnewline
\hline
\multirow{ 4}{*}{$\{(x_i,\{y_3,y_4,y_6,y_7 \}), i \in \lceil 4 \rfloor  \} $} & $c_6+y_6+c_4+y_4+c_3+y_3$ & \multirow{4}{*}{$\{(x_i,\{y_3,y_5,y_6,y_7 \}), i \in \lceil 4 \rfloor  \} $} & $c_7+y_7+c_5+y_5+c_3+y_3$\tabularnewline
& $c_7+y_7+c_6+y_6+c_3+y_3$ & & $c_7+y_7+c_6+y_6+c_3+y_3$ \tabularnewline
& $c_3+y_3$ & & $c_3+y_3$ \tabularnewline
& $c_4+y_4$ & & $c_7+y_7+c_6+y_6+c_5+y_5$ \tabularnewline
\hline
\end{tabular}
\label{Tab:DecodeHamming1}
\end{table*}

\begin{table*}[ht]
\centering{}
\caption{Decoding procedure for receivers in $R_{2}$ of Example \ref{Eg:Hamming}.}
\begin{tabular}{|c|c|c|c|}
\hline
Circuits & Receivers & Decoding Process \tabularnewline
\hline
\multirow{ 4}{*}{ $\lbrace y_1,y_2,y_4,y_7 \rbrace  $ } & $\{ (y_1,\{y_2+y_4+y_7 \})$ & $y_7+y_4+y_2+c_7+c_4+c_2+c_1$\tabularnewline
& $\{ (y_2,\{y_1+y_4+y_7 \})$ &  $y_7+y_4+y_1+c_7+c_4+c_1+c_2$ \tabularnewline
& $\{ (y_4,\{y_1+y_2+y_7 \})$ &  $y_1+y_2+y_7+c_7+c_4+c_2+c_1$ \tabularnewline
& $\{ (y_7,\{y_1+y_2+y_4 \})$ &  $y_1+y_2+y_4+c_7+c_4+c_2+c_1$ \tabularnewline
\hline
\multirow{ 4}{*}{ $\lbrace y_1,y_2,y_5,y_6 \rbrace  $ } & $\{ (y_1,\{y_2+y_5+y_6 \})$ & $y_2+y_5+y_6+c_6+c_5+c_2+c_1$\tabularnewline
& $\{ (y_2,\{y_1+y_5+y_6 \})$ &  $y_1+y_5+y_6+c_6+c_5+c_2+c_1$ \tabularnewline
& $\{ (y_5,\{y_1+y_2+y_6 \})$ &  $y_1+y_2+y_6+c_6+c_5+c_2+c_1$ \tabularnewline
& $\{ (y_6,\{y_1+y_2+y_5 \})$ &  $y_1+y_2+y_5+c_6+c_5+c_2+c_1$ \tabularnewline
\hline
\multirow{ 4}{*}{ $\lbrace y_1,y_3,y_4,y_6 \rbrace  $ } & $\{ (y_1,\{y_3+y_4+y_6 \})$ & $y_3+y_4+y_6+c_6+c_4+c_3+c_1$\tabularnewline
& $\{ (y_3,\{y_1+y_4+y_6 \})$ &  $y_1+y_4+y_6+c_6+c_4+c_3+c_1$ \tabularnewline
& $\{ (y_4,\{y_1+y_3+y_6 \})$ &  $y_1+y_3+y_6+c_6+c_4+c_3+c_1$ \tabularnewline
& $\{ (y_6,\{y_1+y_3+y_4 \})$ &  $y_1+y_3+y_4+c_6+c_4+c_3+c_1$ \tabularnewline
\hline
\multirow{ 4}{*}{ $\lbrace y_1,y_3,y_5,y_7 \rbrace  $ } & $\{ (y_1,\{y_3+y_5+y_7 \})$ & $y_3+y_5+y_7+c_7+c_5+c_3+c_1$\tabularnewline
& $\{ (y_3,\{y_1+y_5+y_7 \})$ &  $y_1+y_5+y_7+c_7+c_5+c_3+c_1$ \tabularnewline
& $\{ (y_5,\{y_1+y_3+y_7 \})$ &  $y_3+y_3+y_7+c_7+c_5+c_3+c_1$ \tabularnewline
& $\{ (y_7,\{y_1+y_3+y_5 \})$ &  $y_1+y_3+y_5+c_7+c_5+c_3+c_1$ \tabularnewline
\hline
\multirow{ 4}{*}{ $\lbrace y_2,y_3,y_4,y_5 \rbrace  $ } & $\{ (y_2,\{y_3+y_4+y_5 \})$ & $y_3+y_4+y_5+c_5+c_4+c_3+c_2$\tabularnewline
& $\{ (y_3,\{y_2+y_4+y_5 \})$ &  $y_2+y_4+y_5+c_5+c_4+c_3+c_2$ \tabularnewline
& $\{ (y_4,\{y_2+y_3+y_5 \})$ &  $y_2+y_3+y_5+c_5+c_4+c_3+c_2$ \tabularnewline
& $\{ (y_5,\{y_2+y_3+y_4 \})$ &  $y_2+y_3+y_4+c_5+c_4+c_3+c_2$ \tabularnewline
\hline
\multirow{ 4}{*}{ $\lbrace y_2,y_3,y_6,y_7 \rbrace  $ } & $\{ (y_2,\{y_3+y_6+y_7 \})$ & $y_2+y_6+y_7+c_7+c_6+c_3+c_2$\tabularnewline
& $\{ (y_3,\{y_2+y_6+y_7 \})$ &  $y_2+y_4+y_5+c_5+c_4+c_3+c_2$ \tabularnewline
& $\{ (y_6,\{y_2+y_3+y_7 \})$ &  $y_2+y_3+y_5+c_5+c_4+c_3+c_2$ \tabularnewline
& $\{ (y_7,\{y_2+y_3+y_6 \})$ &  $y_2+y_3+y_4+c_5+c_4+c_3+c_2$ \tabularnewline
\hline
\multirow{ 4}{*}{ $\lbrace y_4,y_5,y_6,y_7 \rbrace  $ } & $\{ (y_4,\{y_5+y_6+y_7 \})$ & $y_5+y_6+y_7+c_7+c_6+c_5+c_4$\tabularnewline
& $\{ (y_5,\{y_4+y_6+y_7 \})$ &  $y_4+y_6+y_7+c_7+c_6+c_5+c_4$ \tabularnewline
& $\{ (y_6,\{y_4+y_5+y_7 \})$ &  $y_4+y_5+y_7+c_7+c_6+c_5+c_4$ \tabularnewline
& $\{ (y_7,\{y_4+y_5+y_6 \})$ &  $y_4+y_5+y_6+c_7+c_6+c_5+c_4$ \tabularnewline
\hline
\end{tabular}
\label{Tab:DecodeHamming2}
\end{table*}

\begin{table}[ht]
\centering{}
\caption{Decoding procedure for receivers in $R_{3}$ of Example \ref{Eg:Hamming}.}
\begin{tabular}{|c|c|}
\hline
Receivers in $R_{3}$& Decoding Procedure \tabularnewline
\hline
\multirow{ 7}{*}{$\lbrace (y_{i}, \lbrace x_{1},x_{2},x_3,x_4 \rbrace) i \in \lceil 7 \rfloor \rbrace  $}  & $y_1=x_1+c_1$ \tabularnewline
& $y_2=x_2+c_2$ \tabularnewline
& $y_3=x_3+c_3$ \tabularnewline
& $y_4=x_4+c_4$ \tabularnewline
& $ y_5=x_2+x_3+x_4+c_5$ \tabularnewline
& $ y_6=x_1+x_3+x_4+c_6$ \tabularnewline
& $y_7=x_1+x_2+x_4+c_7$ \tabularnewline
\hline
\end{tabular}
\label{Tab:DecodeHamming3}
\end{table}

\end{example}

%\begin{example}
%
%Consider the set of matrices 
%
%\small
% \begin{equation*}
%   M_1=\begin{bmatrix}
%          1 & 0 \\
%          0 & 1 \\
%          0 & 0 \\
%          0 & 0 \\
%        \end{bmatrix},
%    M_2=\begin{bmatrix}
%          0 & 0 \\
%          0 & 0 \\
%          1 & 0 \\
%          0 & 1 \\
%        \end{bmatrix},
%        M_3=\begin{bmatrix}
%          1 & 0 \\
%          0 & 1 \\
%          0 & 1 \\
%          1 & 0 \\
%        \end{bmatrix},
%      M_4=\begin{bmatrix}
%          1 & 0 \\
%          0 & 1 \\
%          1 & 0 \\
%          1 & 1 \\
%        \end{bmatrix}.
% \end{equation*}
%
%\normalsize
%The set of matrices forms a multilinear representation of dimension $2$ for the matroid $U_{2,4}$ over the binary field. However the index coding problem with coded side information constructed from the matroid does not have a vector linear solution of dimension two. The index coding with coded side information problem corresponding to this matroid has the source messages set $\lbrace y_{1},y_{2},y_{3},x_{1},x_{2} \rbrace$,  each message belongs to the finite field $\mathbb{F}_{2}^{2}$. Each message $y_{i}=(y_{i,1},y_{i,2})$ where $y_{i,1},y_{i,2} \in \mathbb{F}_{2}$. Similarly $x_{i}=(x_{i,1},x_{i,2})$ where $x_{i,1},x_{i,2} \in \mathbb{F}_{2}$. Consider receiver $R_{i}$ demanding $y_{1}$ and possessing the coded side information $y_{2}+y_{3}$.  
%\end{example}

\section{Conclusion}
\label{Sec:Conclusion}
In this work we establish few  connections between generalized index coding and discrete polymatroids. It is shown that the existence of a linear solution for a generalized index coding problem is connected to the existence of a representable discrete polymatroid satisfying certain conditions determined by the generalized index coding problem. From a discrete polymatroid a corresponding generalized index coding problem is constructed and it was shown that a representation to the discrete polymatroid exists if a perfect vector linear solution exists for the generalized index coding problem. An example is provided in the paper to illustrate that the converse of the above result is not true. When a similar generalized index coding problem is constructed from the matroids we show that a binary representation to the matroid exists if and only if the constructed index coding problem has a binary scalar linear solution. The connection is helpful in determining whether the index coding problem has a perfect binary scalar linear solution. 

The results of this paper could be extended in the following directions. The construction explained in Section \ref{Sec:GICfromDPM} is general and can be applied to any discrete polymatroid. A generalized index coding problem can be constructed from a non representable discrete polymatroid and further connections could be explored. Also for the constructed index coding problem, certain receivers (belonging to the set $R_2$) possesses the sum of certain elements as its \textit{Has-set}. The elements of the \textit{Has-set} can be made into any other linear combinations and further study could be done. Similar extensions can be considered to Theorem \ref{th:GICfromMat}. Connections between the matroids representable over non binary field and the generalized index coding problems constructed out of those matroids could also be explored.

\section*{Acknowledgement}
This work was supported partly by the Science and Engineering Research Board (SERB) of Department of Science and Technology (DST), Government of India, through J.C. Bose National Fellowship to B. Sundar Rajan.

\end{document}